\documentclass[runningheads]{llncs}
\usepackage{etoolbox}
\newtoggle{TR}
\toggletrue{TR}
\makeatletter
\RequirePackage[bookmarks,unicode,colorlinks=true]{hyperref}%
\def\@citecolor{blue}%
\def\@urlcolor{blue}%
\def\@linkcolor{blue}%

\def\orcidID#1{\smash{\href{http://orcid.org/#1}{\protect\raisebox{-1.25pt}{\protect\includegraphics{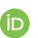}}}}}
\makeatother
\usepackage{graphicx}
\newtoggle{showtodos}
\togglefalse{showtodos}

\usepackage[T1]{fontenc}
\usepackage[utf8x]{inputenx}
\usepackage[american]{babel}
\usepackage{tikz}
\usetikzlibrary{decorations,arrows,shapes,automata,calc}

\definecolor{gencolor}{rgb}{0,0,1}
\definecolor{pscolor}{rgb}{0,1,0}

\tikzstyle{point}=[,circle,fill,minimum size=3pt, inner sep=0pt]
\tikzstyle{dirvec}=[densely dotted,->]
\tikzstyle{nonach}=[fill=red!70]
\tikzstyle{ach}=[fill=pscolor!60!white]
\tikzstyle{dk}=[fill=gray]
\tikzstyle{genach}=[fill=gencolor!30!white]
\tikzstyle{genborder}=[draw=gencolor, densely dotted,very thick]
\tikzstyle{border}=[draw=blue]
\tikzstyle{auxline}=[densely dotted, line width=.8pt]
\tikzstyle{empty}=[text centered, text width=15mm]

\usepackage{amsmath}
\usepackage{thmtools, thm-restate}
\usepackage{xspace}
\usepackage{array}
\usepackage{enumitem}
\usepackage{upgreek}
\usepackage{multirow}
\usepackage{tabularx}
\usepackage{todonotes}
\usepackage{amssymb}
\usepackage{environ}
\usepackage{cleveref}
\usepackage{caption}
\usepackage{subcaption}
\usepackage{stmaryrd}
\usepackage{adjustbox}
\usepackage{makecell}
\usepackage{pgfplots}
\usepackage{pgfplotstable}
\newcommand{\reftr}[1]{\iftoggle{TR}{\Cref{app:#1}}{\cite[App.~#1]{tr}}}

\newcolumntype{?}{!{\vrule width 1.5pt}}
\newcommand{\thickhline}{\Xhline{2\arrayrulewidth}}

\makeatletter
\newcolumntype{\expand}{}
\long\@namedef{NC@rewrite@\string\expand}{\expandafter\NC@find}

\NewEnviron{nproblem}[2][]{%

	\def\problem@arg{#1}%
	\def\problem@framed{framed}%
	\def\problem@lined{lined}%
	\def\problem@doublelined{doublelined}%
	\ifx\problem@arg\@empty%
	\def\problem@hline{}%
	\else%
	\ifx\problem@arg\problem@doublelined%
	\def\problem@hline{\hline\hline}%
	\else%
	\def\problem@hline{\hline}%
	\fi%
	\fi%
	\ifx\problem@arg\problem@framed%
	\def\problem@tablelayout{|>{\bfseries}lX|c}%
	\def\problem@title{\multicolumn{2}{|l|}{%
			\raisebox{-\fboxsep}{\textsc{#2}}%
	}}%
	\else
	\def\problem@tablelayout{>{\bfseries}lXc}%
	\def\problem@title{\multicolumn{2}{l}{%
			\raisebox{-\fboxsep}{\textsc{#2}}%
	}}%
	\fi%
	\goodbreak 
	\medskip\par\noindent%
	\begin{tabularx}{\textwidth}{\expand\problem@tablelayout}%
		\problem@hline%
		\problem@title\\[2\fboxsep]%
		\BODY\\\problem@hline%
	\end{tabularx}%
	\medskip\par
	\goodbreak 
}
\makeatother

\declaretheorem[name=Restriction]{restr}

\Crefname{figure}{Fig.}{Figs.}
\crefname{figure}{fig.}{figs.}
\Crefname{tabular}{Tab.}{Tabs.}
\crefname{tabular}{tab.}{tabs.}
\Crefname{section}{Sect.}{Sects.}
\crefname{section}{sect.}{sects.}
\crefname{restr}{restriction}{restrictions}
\Crefname{restr}{Restriction}{Restrictions}

\setitemize{noitemsep,topsep=0pt,parsep=0pt,partopsep=0pt,leftmargin=11.0pt}
\setenumerate{noitemsep,topsep=0pt,parsep=0pt,partopsep=0pt,leftmargin=12.5pt}

\newcounter{lplinecounter}
\newenvironment{lp}{\setlength{\abovedisplayskip}{-4pt}\setlength{\belowdisplayskip}{-1pt}\begin{equation*}\begin{array}{lrlr}\multicolumn{4}{r}{\hspace{\textwidth}~}\\[-12pt]}{\end{array}\end{equation*}}
\newcommand{\lpcomment}[1]{{\footnotesize \triangleright\, \textit{#1}}}
\newcommand{\lplineintern}[5]{%
	\multicolumn{\ifstreq{#3#4#5}{}{4}{\ifstreq{#3#4}{}{3}{\ifstreq{#3}{}{2}{1}}}}{#1l}{\displaystyle#2}
	\ifstreq{#3}{}{}{&
		\ifstreq{#4}{}{\multicolumn{\ifstreq{#3#5}{}{3}{2}}{l}{\displaystyle#3}}
		{\displaystyle#3}
	}%
	\ifstreq{#4}{}{}{&\ifstreq{#3#5}{}{\multicolumn{2}{l}{\displaystyle#4}}{\displaystyle#4}}
	\ifstreq{#3#5}{}{}{&}
	\ifstreq{#5}{}{}{\enlargecolumn\lpcomment{#5}}
	\ifstreq{#3}{}{}{\refstepcounter{lplinecounter}\ifstreq{#5}{}{}{~~}(\arabic{lplinecounter})} 
}
\makeatletter
\newcommand{\ltxlabel}{\ltx@label}
\makeatother
\newcommand{\lpline}[6]{\lplineintern{#1}{#2}{#3}{#4}{#5}\ltxlabel{{#6}}\\}

\newcommand{\enlargecolumn}{\hspace{-8cm}~}
\crefname{lplinecounter}{constraint}{constraints}
\Crefname{lplinecounter}{Constraint}{Constraints}

\newcommand{\tool}[1]{\textsc{#1}\xspace}
\newcommand{\benchmark}[1]{\textsf{#1}\xspace}
\newcommand{\eg}{e.g.,\xspace}
\newcommand{\ie}{i.e.,\xspace}

\newcommand{\iverson}[1]{\ensuremath{\left[ #1 \right]}}
\newcommand{\tuple}[1]{\ensuremath{\left\langle #1 \right\rangle}}

\newcommand{\tupleaccess}[2]{{\ensuremath{#1\llbracket#2\rrbracket}}}
\newcommand{\set}[1]{\ensuremath{\left\{ #1 \right\}}}

\newcommand{\rr}{\ensuremath{\mathbb{R}}}
\newcommand{\nn}{\ensuremath{\mathbb{N}}}
\newcommand{\qq}{\ensuremath{\mathbb{Q}}}
\newcommand{\zz}{\ensuremath{\mathbb{Z}}}
\newcommand{\ex}[1]{\ensuremath{\exists\,#1\colon\,}}
\newcommand{\fa}[1]{\ensuremath{\forall\,#1\colon\,}}
\newcommand{\dist}{\ensuremath{\mu}}
\newcommand{\dists}[1]{\ensuremath{\mathit{Dist(#1)}}}
\newcommand{\supp}[1]{\ensuremath{\mathit{supp(#1)}}}

\DeclareMathOperator*{\argmax}{arg\,max}

\newcommand{\mdp}{\ensuremath{\mathcal{M}}}
\newcommand{\states}{\ensuremath{S}}
\newcommand{\actions}{\ensuremath{\mathit{Act}}}
\newcommand{\transitions}{\ensuremath{\mathbf{P}}}
\newcommand{\sinit}{\ensuremath{s_{\mathit{I}}}}
\newcommand{\mdptuple}{\ensuremath{ \tuple{\states, \actions, \transitions, \sinit} }}
\newcommand{\state}{\ensuremath{s}}
\newcommand{\action}{\ensuremath{\alpha}}
\newcommand{\actbot}{\ensuremath{\bot}}
\newcommand{\actinit}{\ensuremath{\action_{\mathit{I}}}}
\newcommand{\act}[1]{\ensuremath{\mathit{Act}\ifthenelse{\equal{#1}{}}{}{(#1)}}}
\newcommand{\successors}[2]{\ensuremath{\mathit{succ}(#1,#2)}}
\newcommand{\predecessors}[1]{\ensuremath{\mathit{pre}(#1)}}
\newcommand{\stateactionpairset}{\ensuremath{\mathcal{E}}}
\newcommand{\submdp}[2]{\ensuremath{\tupleaccess{#1}{#2}}}

\newcommand{\mecs}[1]{\ensuremath{\mathit{MECS(#1)}}}

\newcommand{\infpath}{\ensuremath{\pi}}
\newcommand{\finpath}{\ensuremath{\hat{\pi}}}
\newcommand{\last}[1]{\ensuremath{\mathit{last}(#1)}}
\newcommand{\infpaths}[1]{\ensuremath{\mathit{Paths}_\mathrm{inf}^{#1}}}
\newcommand{\finpaths}[1]{\ensuremath{\mathit{Paths}_\mathrm{fin}^{#1}}}
\newcommand{\paths}{\ensuremath{\Pi}}
\newcommand{\stateofpath}[2]{\ensuremath{#1[#2]}}
\newcommand{\lengthofpath}[1]{\ensuremath{|#1|}}

\newcommand{\sched}{\ensuremath{\sigma}}
\newcommand{\scheds}[1]{\ensuremath{\Sigma^{#1}}}

\newcommand{\psscheds}[1]{\ensuremath{\Sigma_{\mathrm{PS}}^{#1}}}

\newcommand{\memorystates}{\ensuremath{M}}
\newcommand{\memorystate}{\ensuremath{m}}
\newcommand{\schednextaction}{\ensuremath{\sched_{a}}}
\newcommand{\schedmemoryupdate}{\ensuremath{\sched_{u}}}
\newcommand{\schedtuple}{\ensuremath{\tuple{\memorystates,\schednextaction,\schedmemoryupdate,\minit}}}
\newcommand{\pschedsK}[2]{\ensuremath{\scheds{#1}_{\mathrm{P}, #2}}}
\newcommand{\pKach}[3]{\ensuremath{\mathit{Ach}_{\mathrm{P}, #3}^{#1}(#2)}}
\newcommand{\pKpareto}[3]{\ensuremath{\mathit{Pareto}_{\mathrm{P}, #3}^{#1}(#2)}}

\newcommand{\mc}[2]{\ensuremath{#1^{#2}}}
\newcommand{\mctransitions}[1]{\ensuremath{\transitions^{#1}}}

\newcommand{\event}{\ensuremath{\Pi}}
\newcommand{\mdpwithinitstate}[2]{\ensuremath{#1_{#2}}}

\newcommand{\ndmemory}[1]{\ensuremath{\mathcal{N}_{#1}}}
\newcommand{\ndmemoryupdate}{\ensuremath{\delta}}
\newcommand{\minit}{\ensuremath{m_{\mathit{I}}}}
\newcommand{\ndmemorytuple}{\ensuremath{\tuple{\memorystates, \ndmemoryupdate, \minit}}}
\newcommand{\mdpndmemory}[1]{\ensuremath{\mdp \otimes {\ndmemory{#1}}}}
\newcommand{\mdpndmemorystates}{\ensuremath{\states'}}
\newcommand{\mdpndmemoryactions}{\ensuremath{\actions'}}
\newcommand{\mdpndmemorytransitions}{\ensuremath{\transitions'}}
\newcommand{\mdpndmemorysinit}{\ensuremath{\sinit'}}
\newcommand{\mdpndmemorytuple}{\ensuremath{\tuple{\mdpndmemorystates, \mdpndmemoryactions, \mdpndmemorytransitions, \mdpndmemorysinit}}}

\newcommand{\probmeasure}[2]{\ensuremath{\mathrm{Pr}^{#1}_{#2}}}
\newcommand{\expval}[2]{\ensuremath{\mathrm{E}^{#1}_{#2}}}

\newcommand{\eventually}{\ensuremath{\lozenge}}
\newcommand{\goalstates}{\ensuremath{G}}
\newcommand{\goalstatesj}{\goalstates_\objindex}

\newcommand{\reachobjop}{\ensuremath{\mathbb{P}}}

\newcommand{\rewstruct}{\ensuremath{\mathbf{R}}}
\newcommand{\rewstructj}{\ensuremath{\rewstruct_\objindex}}
\newcommand{\rewarddomain}{\ensuremath{\rr_{\ge 0}}}
\newcommand{\rewofpath}[2]{\ensuremath{#1(#2)}}

\newcommand{\rewobjop}{\ensuremath{\mathbb{E}}}

\newcommand{\pointi}{\ensuremath{p}}
\newcommand{\point}{\ensuremath{\vec{p}}}
\newcommand{\points}{\ensuremath{P}}
\newcommand{\pointidomain}{\ensuremath{\rr_\infty}}
\newcommand{\pointdomain}{\ensuremath{(\pointidomain)^\numobj}}

\newcommand{\obj}{\ensuremath{\psi}}
\newcommand{\objj}{\ensuremath{\obj_\objindex}}
\newcommand{\objrel}{\ensuremath{\sim}}
\newcommand{\objrelj}{\ensuremath{\objrel_\objindex}}

\newcommand{\numobj}{\ensuremath{\ell}}

\newcommand{\objindex}{\ensuremath{j}}

\newcommand{\multiobjquery}{\ensuremath{\mathcal{Q}}}
\newcommand{\ach}[2]{\ensuremath{\mathit{Ach}^{#1}(#2)}}

\newcommand{\psach}[2]{\ensuremath{\mathit{Ach}_{\mathrm{PS}}^{#1}(#2)}}
\newcommand{\pareto}[2]{\ensuremath{\mathit{Pareto}^{#1}(#2)}}

\newcommand{\pspareto}[2]{\ensuremath{\mathit{Pareto}_{\mathrm{PS}}^{#1}(#2)}}
\newcommand{\closure}[2]{\ensuremath{\mathit{cl}^{#1}(#2)}}

\newcommand{\precision}{\ensuremath{\vec{\epsilon}}}
\newcommand{\precisiondomain}{\ensuremath{(\rr_{>0})^\numobj}}
\newcommand{\underapprox}{\ensuremath{L}}
\newcommand{\overapprox}{\ensuremath{U}}
\newcommand{\approxtuple}{\ensuremath{\tuple{\underapprox, \overapprox}}}
\newcommand{\region}{\ensuremath{\mathcal{R}}}
\newcommand{\dirvec}[1][]{\ensuremath{\vec{w}_{#1}}}

\newcommand{\problemabbreviation}[1]{\textup{\texttt{#1}}\xspace}
\newcommand{\gma}{\problemabbreviation{GMA}}
\newcommand{\psma}{\problemabbreviation{PSMA}}
\newcommand{\pbma}{\problemabbreviation{PBMA}}

\newcommand{\pspapprox}{\problemabbreviation{PSP}\!\ensuremath{^\approx}\xspace}
\newcommand{\pbpapprox}{\problemabbreviation{PBP}\!\ensuremath{^\approx}\xspace}
\newcommand{\milp}{\problemabbreviation{MILP}}

\newcommand{\lpactionvar}[2]{\ensuremath{a_{#1,#2}}}
\newcommand{\lpactionvarsa}{\lpactionvar{\state}{\action}}
\newcommand{\lpecvar}[2]{\ensuremath{e^{#2}_{#1}}}
\newcommand{\lpecvarsj}{\lpecvar{\state}{\objindex}}
\newcommand{\lpecvars}{\lpecvar{\state}{}}
\newcommand{\lpecactionvar}[3]{\ensuremath{e^{#3}_{#1,#2}}}
\newcommand{\lpecactionvarsaj}{\lpecactionvar{\state}{\action}{\objindex}}
\newcommand{\lpecflowvar}[3]{\ensuremath{z^{#3}_{#1,#2}}}
\newcommand{\lpecflowvarsaj}{\lpecflowvar{\state}{\action}{\objindex}}
\newcommand{\maxflow}[1]{V_{#1}}
\newcommand{\pmobj}{\pm}
\newcommand{\minobj}{\iverson{\mathrm{min}}}
\newcommand{\lpvalvar}[2]{\ensuremath{x^{#2}_{#1}}}
\newcommand{\lpvalvarsj}{\lpvalvar{\state}{\objindex}}
\newcommand{\lpvalactionvar}[3]{\ensuremath{x^{#3}_{#1,#2}}}
\newcommand{\lpvalactionvarsaj}{\lpvalactionvar{\state}{\action}{\objindex}}
\newcommand{\lpflowvar}[2]{\ensuremath{y_{#1,#2}}}
\newcommand{\lpflowvarsa}{\lpflowvar{\state}{\action}}
\newcommand{\lpflowecvar}[1]{\ensuremath{\lpflowvar{#1}{\actbot}}}
\newcommand{\lpflowecvars}{\lpflowecvar{\state}}

\newcommand{\var}{\ensuremath{\mathit{Var}}}
\newcommand{\lpsol}{\ensuremath{\Phi}}
\newcommand{\lpsolof}[1]{\ensuremath{\lpsol(#1)}}

\newcommand{\actec}[1]{\ensuremath{\submdp{\actions}{\stateactionpairset}\ifthenelse{\equal{#1}{}}{}{(#1)}}}
\newcommand{\uppervaluebound}[2]{U_{#1}^{#2}}
\newcommand{\uppervalueboundsj}{\uppervaluebound{\state}{\objindex}}
\newcommand{\stateszeroj}{\states_0^\objindex}
\newcommand{\stateszero}{\states_0}
\newcommand{\statesinf}{\states_\infty}

\newcommand{\statesmaybej}{\states_?^j}
\newcommand{\statesmaybe}{\states_?}

\newcommand{\timeout}{\texttt{TO}}
\newcommand{\memout}{\texttt{MO}}
\newcommand{\error}{\texttt{ERR}}

\makeatletter
\newcommand{\ifstreq}[2]{%
  \ifnum\pdfstrcmp{\detokenize{#1}}{\detokenize{#2}}=\z@
\expandafter\@firstoftwo
\else
\expandafter\@secondoftwo
\fi}
\makeatother
\makeatletter
\newcolumntype{"}{@{\ }@{\hskip\tabcolsep\vrule width 2pt\hskip\tabcolsep}@{\ }}
\newcolumntype{R}{>{$}r<{$}}
\makeatother

\newcommand{\tq}[1]{\iftoggle{showtodos}{\todo[color=red!30]{TQ: #1}}{}}

\hypersetup{
  pdftitle = {Simple Strategies in Multi-Objective MDPs}
}

\iftoggle{TR}{}{
\usepackage{graphicx}
\usepackage[firstpage]{draftwatermark}
\SetWatermarkText{\hspace*{10.5cm}\raisebox{15.5cm}{\includegraphics{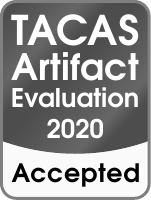}}}
\SetWatermarkAngle{0}
}

\begin{document}
\title{%
Simple Strategies in Multi-Objective MDPs%
\thanks{Research partially supported by F.R.S.-FNRS Grant n\textsuperscript{$\circ$} F.4520.18 (\emph{ManySynth}).\ Mickael Randour is an F.R.S.-FNRS Research Associate.}
{}}
\iftoggle{TR}{\subtitle{Technical Report}}{}%
\author{
Florent Delgrange\inst{1,2}\fnmsep\protect\footnote[4]{currently affiliated with Vrije Universiteit Brussel.}\orcidID{0000-0003-2254-0596} 
\and Joost-Pieter Katoen\inst{1}\orcidID{0000-0002-6143-1926}\and\\
Tim Quatmann\inst{1}\orcidID{0000-0002-2843-5511}
\and Mickael~Randour\inst{2}
}
\authorrunning{F. Delgrange et al.}
\institute{
RWTH Aachen University, Aachen, Germany
\and UMONS -- Université de Mons, Mons, Belgium
}
\date{\today}
\maketitle

\begin{abstract}
We consider the verification of multiple expected reward objectives at once on Markov decision processes (MDPs).
This enables a trade-off analysis among multiple objectives by obtaining a Pareto front.
We focus on strategies that are easy to employ and implement.
That is, strategies that are pure (no randomization) and have bounded memory.
We show that checking whether a point is achievable by a pure stationary strategy is NP-complete, even for two objectives, and we provide an \milp encoding to solve the corresponding problem.
The bounded memory case is treated by a product construction.
Experimental results using \tool{Storm} and \tool{Gurobi} show the feasibility of our algorithms.
\end{abstract}


\section{Introduction}
\label{sec:Introduction}

\textit{MDPs.} Markov decision processes (MDPs)~\cite{BK08,BHK19} are a key model in stochastic decision making. The classical setting involves a system subject to a stochastic model of its environment, and the goal is to synthesize a system controller, represented as a \textit{strategy} for the MDP, ensuring a given level of \textit{expected performance}. Tools such as \tool{Prism}~\cite{KNP11} and \tool{Storm}~\cite{DJKV17} support MDP model checking. 

\smallskip\noindent\textit{Multi-objective MDPs.} MDPs where the goal is to achieve a \textit{combination} of objectives (rather than just one) are popular in e.g., AI~\cite{RVWD13} and verification~\cite{DBLP:conf/csl/BaierDK14}. This is driven by applications, where controllers have to fulfill multiple, potentially conflicting objectives, requiring a \textit{trade-off} analysis. This includes multi-dimension MDPs~\cite{CMH06,EKVY07,RRS17,CKK17} where weight vectors are aggregated at each step and MDPs where the specification mixes different views (e.g., average and worst~case performance) of the same weight
~\cite{DBLP:journals/iandc/BruyereFRR17,DBLP:conf/icalp/BerthonRR17}.
With multiple objectives, optimal strategies no longer exist in general: instead, \textit{Pareto-optimal} strategies are considered. The Pareto front, i.e., the set of non-dominated achievable value vectors is usually non-trivial. Elaborate techniques are needed to explore it efficiently, e.g.,~\cite{FKNPQ11,FKP12}.

\smallskip\noindent\textit{Simple strategies.} Another stumbling block in multi-objective MDPs is the complexity of strategies: Pareto-optimal strategies typically need both \textit{memory} and \textit{randomization}. A simple conjunction of reachability objectives already requires randomization and exponential memory (in the number of reachability sets)~\cite{RRS17}. Some complex objectives even need infinite memory, e.g.,~\cite{DBLP:journals/iandc/BruyereFRR17,DBLP:conf/icalp/BerthonRR17}.
In controller synthesis, strategies requiring randomization and/or (much) memory may not be practical. Limited-memory strategies are required on devices with limited resources~\cite{BBPM06}.  Randomization is elegant and powerful from a theoretical view, but has practical limitations, e.g., it limits reproducibility which complicates debugging. Randomized strategies are also often despised for medical applications~\cite{LBM12} and product design -- all products should have the same design, not a random one. This motivates to consider the analysis of \textit{simple strategies}, i.e., strategies using no randomization and a limited amount of memory (given as a parameter). While most works study the Pareto front among \textit{all} strategies, we establish ways to explore efficiently the Pareto front among \textit{simple} strategies only.

\smallskip\noindent\textit{Problem statement.} We consider pure (i.e., no randomization) and bounded-memory strategies and study two problems: (a)~\textit{achievability} queries -- is it possible to achieve a given value vector -- and (b)~\textit{approximation of the Pareto front}. Considering pure, bounded-memory strategies is natural as randomization can be traded for memory~\cite{CAH04}: without randomization, optimal strategies may require arbitrarily large memory, (see Ex.~\ref{ex:nontrivial-memory}).
We study mixtures of \textit{expected (accumulated) reward objectives}, covering various studied settings like reachability~\cite{EKVY07,RRS17}, shortest path~\cite{DBLP:conf/vmcai/RandourRS15,RRS17,HJKQ18,DBLP:journals/corr/abs-1809-03107} and total reward objectives~\cite{FKNPQ11,FKP12}.

\smallskip\noindent\textit{Contributions.} 
We first consider the achievability problem for pure stationary (i.e., memoryless) strategies and show that finding optimal strategies for multi-objective MDPs is NP-complete, even for two objectives. This contrasts the case of general strategies, where the problem is polynomial-time if the number of objectives is fixed~\cite{RRS17}. We provide a \textit{mixed integer linear program} (\milp) encoding. The crux lies in dealing with end components. The \milp is polynomial in the input MDP and the number of objectives.  Inspired by~\cite{FKNP11}, we give an alternative \milp encoding which is better suited for total reward objectives. To approximate the Pareto front under pure stationary strategies, we solve multiple \milp queries. This iteratively divides the solution space into achievable and non-achievable regions. Bounded-memory strategies are treated via a product construction. Our approach works for finite \textit{and} infinite expected rewards.

\smallskip\noindent\textit{Practical evaluation.} We successfully compute Pareto fronts for 13 benchmarks using our implementation in \tool{Storm}, exploiting the \milp solver \tool{Gurobi}. Despite the hard nature of the problem, our experiments show that Pareto fronts for models with tens of thousands of states can be successfully approximated. 

\smallskip\noindent\textit{Related work.} NP completeness for discounted rewards under pure strategies was shown in~\cite{CMH06}.
\cite{EG16} claims that this generalizes to PCTL objectives but no proof is given.
%
%
\cite{SBHH17} treats multi-objective bounded MDPs whose transition probabilities are intervals.
A set of Pareto optimal policies is computed using policy iteration and an efficient heuristic is exploited to compute a set of mutually non-dominated policies that are likely to be Pareto optimal. 
 Pure stationary Pareto optimal strategies for discounted rewards are obtained in \cite{WJ07} using value-iteration but is restricted to small MDPs where all probabilities are 0 or 1. In~\cite{PW10},Tchebycheff-optimal strategies for discounted rewards are obtained via an LP approach; such strategies minimize the distance to a reference point and are not always pure.
\section{Preliminaries}
\label{sec:prelim}
For a finite set $\Omega$, let $\dists{\Omega} = \set{ \dist \colon \Omega \to [0,1] \mid \sum_{\omega \in \Omega} \dist(\omega) = 1 }  $ be the set of probability distributions over $\Omega$ with support $\supp{\dist} = \set{ \omega \in \Omega \mid \dist(\omega) > 0  }$.
We write $\rewarddomain = \set{|x| \mid x \in \rr}$ and $\pointidomain = \rr \cup \set{\infty}$ for the non-negative and extended real numbers, respectively.
$\vec{1}^\numobj = \tuple{1, \dots, 1}$ denotes the vector of size $\numobj \in \nn$ with all entries 1.
We just write $\vec{1}$ if $\numobj$ is clear.
Let $\tupleaccess{\point}{i}$ denote the $i^{th}$ entry and $\point \cdot \point'$ the dot product of $\point,\point' \in \pointdomain$. $\point\le\point'$, $\point+\point'$, and $|\point|$ are \mbox{entry-wise}.
%
For Boolean expression $\mathit{cond}$, let $[\mathit{cond}] = 1$ if $\mathit{cond}$ is true and $[\mathit{cond}] = 0$ otherwise.


\subsection{Markov Decision Processes, Strategies, and End Components}

\begin{definition}[Markov decision process \cite{Put94}]
	A \emph{Markov decision process} (MDP) is a tuple $\mdp = \mdptuple$ with finite set of states $\states$, initial state $\sinit \in \states$, finite set of actions $\actions$, and  transition function
	 $\transitions \colon \states \times \actions \times \states \to [0,1]$ with $\sum_{\state' \in \states} \transitions(\state, \action, \state') \in \set{0,1}$ for all $\state \in \states$ and $\action \in \actions$.
\end{definition}
We fix an MDP $\mdp = \mdptuple$.
Intuitively, $\transitions(\state, \action, \state')$ is the probability to take a transition from $\state$ to $\state'$ when choosing action $\action$.
An infinite path in $\mdp$ is a sequence $\infpath = \state_0 \action_1 \state_1 \action_2 \dots \in (\states {\times} \actions)^\omega$ with $\transitions(\state_i, \action_{i+1}, \state_{i+1}) > 0$ for all $i\in \nn$.
We write $\stateofpath{\infpath}{i} = \state_i$ for the $(i{+}1)$th state visited by $\infpath$ and define the length of $\infpath$ as $\lengthofpath{\infpath} = \infty$.
A finite path is a finite prefix $\finpath = \state_0 \action_1 \dots \action_{n} \state_n$ of infinite path $\infpath$, where $\last{\finpath} = \state_n \in \states$, $\lengthofpath{\finpath} = n$ and $\stateofpath{\finpath}{i}=\state_i$ for $i \le n$.
The set of finite (infinite) paths in $\mdp$ is denoted by $\finpaths{\mdp}$ ($\infpaths{\mdp}$).
The \emph{enabled actions} at a state $\state \in \states$ are given by the set $\act{\state} = \set{\action \in \actions \mid \ex{\state' \in \states} \transitions(\state,\action, \state') > 0  }$. We assume $\act{\state} \neq \emptyset$ for all $\state$.
If $|\act{\state}| = 1$ for all $\state \in \states$, $\mdp$ is called a \emph{Markov Chain (MC)}.
We write $\mdpwithinitstate{\mdp}{\state}$ for the MDP obtained by replacing the initial state of $\mdp$ by $\state \in \states$.
For $\state \in \states$ and $\action \in \actions$, we define the set of successor states $\successors{\state}{\action} = \set{\state'  \mid \transitions(\state,\action,\state') > 0}$.
For $\state' \in \states$, the set of predecessor state-action pairs is given by $\predecessors{\state'} = \set{\tuple{\state,\action}  \mid \transitions(\state,\action,\state') > 0}$.
For a set $\stateactionpairset \subseteq \states \times \actions$, we define
$\submdp{\states}{\stateactionpairset} = \set{\state \in \states \mid \ex{\action} \tuple{\state, \action} \in \stateactionpairset}$,
$\submdp{\actions}{\stateactionpairset} =  \set{\action \in \actions \mid \ex{\state} \tuple{\state, \action} \in \stateactionpairset}$, and
$\submdp{\transitions}{\stateactionpairset}(\state,\action,\state') = \iverson{\tuple{\state,\action} \in \stateactionpairset} \cdot \iverson{\state' \in \submdp{\states}{\stateactionpairset}} \cdot \transitions(\state,\action,\state')$.
We say $\stateactionpairset$ is \emph{closed} for $\mdp$ if $\fa{\tuple{\state,\action} \in \stateactionpairset} \action \in \act{\state}$ and $\successors{\state}{\action} \subseteq \submdp{\states}{\stateactionpairset}$.
\begin{definition}[Sub-MDP]
	The \emph{sub-MDP} of $\mdp$, closed $\stateactionpairset \subseteq \states \times \actions$, and $\state \in \submdp{\states}{\stateactionpairset}$ is given by $\submdp{\mdp}{\stateactionpairset, \state} = \tuple{\submdp{\states}{\stateactionpairset}, \submdp{\actions}{\stateactionpairset}, \submdp{\transitions}{\stateactionpairset}, \state}$.
We also write $\submdp{\mdp}{\stateactionpairset}$ for the sub-MDP $\submdp{\mdp}{\stateactionpairset,\state}$ and an arbitrary state $\state \in \submdp{\states}{\stateactionpairset}$.
\end{definition}
\begin{definition}[End Component]
	\label{def:ec}
	A non-empty set $\stateactionpairset \subseteq \states \times \actions$ is an \emph{end component (EC)} of $\mdp$ if $\stateactionpairset$ is closed for $\mdp$ and
	for each pair of states $\state,\state' \in \submdp{\states}{\stateactionpairset}$ there is a finite path $\finpath \in \finpaths{\submdp{\mdp}{\stateactionpairset}}$ with $ \stateofpath{\finpath}{0} = \state$ and $\last{\finpath} = \state'$.
	An EC $\stateactionpairset$ is \emph{maximal}, if there is no other EC $\stateactionpairset'$ with $\stateactionpairset \subsetneq \stateactionpairset'$.
	The set of all maximal end components of $\mdp$ is $\mecs{\mdp}$.
\end{definition}
The maximal ECs of a Markov chain are also called \emph{bottom strongly connected components (BSCCs)}.
A strategy resolves nondeterminism in MDPs:
\begin{definition}[Strategy]
	A (general) \emph{strategy} for MDP $\mdp$ is a function $\sched \colon \finpaths{\mdp} \to \dists{\actions}$ with $\supp{\sched(\finpath)} \subseteq \act{\last{\finpath}}$ for all $\finpath \in \finpaths{\mdp}$.
\end{definition}
Let $\sched$ be a strategy for $\mdp$. Intuitively, $\sched(\finpath)(\action)$ is the probability to perform action $\action$ after observing history $\finpath \in \finpaths{\mdp}$.
%
%
 A strategy is \textit{pure} if all histories are mapped to \textit{Dirac distributions}, i.e., the support is a singleton.
%
A strategy is \emph{stationary} if its decisions only depend on the current state,
\ie $\fa{\finpath, \finpath' \in\finpaths{\mdp}}\last{\finpath} = \last{\finpath'}$ implies
$ \sched(\finpath) = \sched(\finpath')$.
We often assume $\sched \colon \states \to \dists{\actions}$ for stationary and $\sched \colon \states \to \actions$ for pure stationary strategies $\sched$.
Let $\scheds{\mdp}$ and $\psscheds{\mdp}$ be the sets of general and pure stationary strategies, respectively.
A set of paths $\event \subseteq \infpaths{\mdp}$ is \emph{compliant} with $\sched \in \scheds{\mdp}$ if for all $\infpath = \state_0 \action_1  \state_1 \dots \in \event$ and prefixes $\finpath$ of $\infpath$ satisfy $\sched(\finpath)(\action_{\lengthofpath{\finpath}+1}) > 0$.
%
The \emph{induced Markov chain} of $\mdp$ and $\sched \in \psscheds{\mdp}$ is given by $\mc{\mdp}{\sched} = \submdp{\mdp}{\stateactionpairset^\sched,\sinit}$ with $\stateactionpairset^\sched =  \set{ \tuple{\state, \sched(\state)} \mid \state \in \states}$.	 

MDP $\mdp$ and strategy $\sched \in \scheds{\mdp}$ induce a probability measure $\probmeasure{\mdp}{\sched}$ on subsets $\paths \subseteq \infpaths{\mdp}$ given by a standard cylinder set construction~\cite{BK08,FKNP11}.
%
The expected value of $X \colon \infpaths{\mdp} \to \pointidomain$ is 
 $\expval{\mdp}{\sched}(X) = \int_{\infpath} X(\infpath) \mathop{}\!d\probmeasure{\mdp}{\sched}(\set{\infpath})$.
For $\sched \in \psscheds{\mdp}$, $\probmeasure{\mdp}{\sched}$ and $\expval{\mdp}{\sched}$ coincide with the corresponding measures on MC $\mc{\mdp}{\sched}$.
%
\subsection{Objectives}

A reward structure $\rewstruct \colon \states \times \actions \times \states \to \rewarddomain$ assigns non-negative rewards to transitions.
We accumulate rewards on (in)finite paths $\infpath = \state_0 \action_1 \state_1 \action_2 \dots$:
$\rewofpath{\rewstruct}{\infpath} = \sum_{i=1}^{\lengthofpath{\infpath}} \rewstruct(\state_{i-1}, \action_i, \state_i)$.
For a set of goal states $\goalstates \subseteq \states$, let $\rewofpath{\rewstruct \eventually \goalstates}{\infpath} = \rewstruct(\finpath)$, where $\finpath$ is the smallest prefix of $\infpath$ with $\last{\finpath} \in \goalstates$ (or $\finpath = \infpath$ if no such prefix exists).
Intuitively, $\rewofpath{\rewstruct \eventually \goalstates }{\infpath}$ is the reward accumulated on $\infpath$ until a state in $\goalstates$ is reached.
A (reward) \emph{objective} has the form $\rewobjop_\objrel(\rewstruct \eventually \goalstates)$ for ${\objrel} \in \set{{\ge}, {\le}}$.
We write $\tuple{\mdp, \sched, \pointi} \models \rewobjop_\objrel(\rewstruct \eventually \goalstates)$ iff $\expval{\mdp}{\sched}(\rewstruct \eventually \goalstates) \objrel \pointi$, \ie for $\mdp$ and $\sched$, the expected accumulated reward until reaching $G$ is at least (or at most) $\pointi \in \pointidomain$.
We call the objective \emph{maximizing} if ${\objrel} = {\ge}$ and \emph{minimizing} otherwise.
If $\goalstates = \emptyset$ (\ie $\rewofpath{\rewstruct \eventually \goalstates }{\infpath} = \rewofpath{\rewstruct}{\infpath}$ for all paths $\infpath$), we call the objective a \emph{total reward objective}.
Let the reward structure $\rewstruct^\goalstates$ be given by $\rewstruct(\state,\action,\state') = \iverson{\state' \in \goalstates}$.
Then, $\probmeasure{\mdp}{\sched}(\eventually \goalstates)  = \expval{\mdp}{\sched}(\rewstruct^\goalstates \eventually \goalstates)$ for every $\sched\in \scheds{\mdp}$, where 
$\eventually \goalstates \subseteq \infpaths{\mdp}$ denotes the set of paths that visit a state in $\goalstates$.
We use $\reachobjop_\objrel(\eventually \goalstates)$ as a shortened for $\rewobjop_\objrel(\rewstruct^\goalstates \eventually \goalstates)$
and call such an objective a \emph{reachability objective}.

\begin{definition}[Multi-objective query]
	For MDP $\mdp$, an \emph{$\numobj$-dimensional multi-objective query} is a tuple $\multiobjquery = \tuple{\obj_1, \dots, \obj_\numobj}$ of $\numobj$ objectives $\objj = \rewobjop_{\objrelj}(\rewstructj \eventually \goalstatesj)$.
\end{definition}
Each objective $\objj$ considers a different reward structure $\rewstructj$.
The MDP $\mdp$, strategy $\sched$, and point $\point \in \pointdomain$ satisfy a multi-objective query $\multiobjquery = \tuple{\obj_1, \dots, \obj_\numobj}$ (written $\tuple{\mdp, \sched, \point} \models \multiobjquery$) iff $\fa{\objindex} \tuple{\mdp, \sched, \tupleaccess{\point}{\objindex}} \models \obj_\objindex$.
Then, we also say $\sched$ \emph{achieves} $\point$ and call $\point$ \emph{achievable}.
Let $\ach{\mdp}{\multiobjquery}$ ($\psach{\mdp}{\multiobjquery}$) denote the set of points achieved by a general (pure stationary) strategy.
%
%
The \emph{closure} of a set $\points \subseteq \pointdomain$ with respect to query $\multiobjquery$ is
$\closure{\multiobjquery}{\points} = \set{ \point \in \pointdomain \mid \ex{\point' \in P}\fa{\objindex} \tupleaccess{\point'}{\objindex} \objrelj \tupleaccess{\point}{\objindex} }$.
For $\point, \point' \in \pointdomain$, we say that $\point$ \emph{dominates} $\point'$ if $\point' \in \closure{\multiobjquery}{\set{\point}}$.
In this case, $\tuple{\mdp, \sched, \point} \models \multiobjquery$ implies $\tuple{\mdp, \sched, \point'} \models \multiobjquery$ for any $\sched \in \scheds{\mdp}$.
We are interested in the Pareto front, which is the set of non-dominated achievable points.
\begin{definition}[Pareto front]
	The \emph{(general) Pareto front} for $\mdp$ and $\multiobjquery$ is
	 $\pareto{\mdp}{\multiobjquery} = \set{\point \in \ach{\mdp}{\multiobjquery} \mid \fa{\point' \in \ach{\mdp}{\multiobjquery}} \point \in \closure{\multiobjquery}{\{\point'\}} \implies \point = \point'}$.
\end{definition}
The Pareto front is the smallest set $\points \subseteq \pointdomain$ with $\closure{\multiobjquery}{\points} = \ach{\mdp}{\multiobjquery}$.
In a similar way, we define the \emph{pure stationary} Pareto front $\pspareto{\mdp}{\multiobjquery}$ which only consider points in $\psach{\mdp}{\multiobjquery}$.

\begin{figure}[t]
	\centering
	\begin{subfigure}[b]{0.54\linewidth}
		\scalebox{1}{
		\centering
		\begin{tikzpicture}[yscale=0.6,->,>=stealth',shorten >=1pt,auto,node
		distance=2.5cm,bend angle=45, scale=0.6, font=\normalsize]
		\tikzstyle{p1}=[draw,circle,text centered,minimum size=6mm,text width=3mm]
		\tikzstyle{p2}=[draw,rectangle,text centered,minimum size=6mm,text width=3mm]
		\tikzstyle{act}=[fill,circle,inner sep=1pt,minimum size=1.5pt, node distance=1cm]    \tikzstyle{empty}=[text centered, text width=15mm]
		\node[p1] (1) at (0,0) {$\state_1$};
		\node[p1] (2) at ($(1) - (0,5)$) {$\state_2$};
		\node[p1, accepting] (3) at ($(1) + (5, -2.5)$)  {$\state_4$};
		\node[p1] (4) at ($(1) + (7, 0)$) {$\state_5$};
		\node[p1, accepting] (5) at ($(4) + (0,-5)$) {$\state_6$};
		\node[p2] (6) at ($(1)+(2,-2.5)$) {$\state_3$};
		\node[act] (1a) at ($(1) + (0,-2.5)$) {};
		\node[act] (1b) at (5,0) {};
		\node[act] (2a) at ($(2)+(3,0)$) {};
		\node[act] (2b) at ($(2)+(1,1.25)$) {};
		\node[act] (3a) at ($(3) + (-1.5,0)$) {};
		\node[act] (4a) at ($(4.east)+(1,0)$) {};
		\node[act] (5a) at ($(5)+(0,1.8)$) {};
		\node[act] (6a) at ($(6.north)+(0,1)$) {};
		\node[empty] at ($(1a)-(0.4,0)$) {\scriptsize $\alpha$};
		\node[empty] at ($(1b)+(0,0.4)$) {\scriptsize $\beta$};
		\node[empty] at ($(2a)-(0,0.5)$) {\scriptsize $\gamma$};
		\node[empty] at ($(2b)+(0,0.4)$) {\scriptsize $\delta$};
		\node[]  (init) at (-1.6,0) {};
		\path[-latex']
		(init) edge (1)
		(1) edge (1a);
		\draw[->] (1a) -- (2);
		\draw[->] (1) -- (1b);
		\draw[->] (1b) -- node[above]{\scriptsize $0.3$} (4);
		\draw[->] (1b) -- node[left,yshift=0mm]{\scriptsize $0.7$} (3);
		\draw[->] (2) to [] (2a);
		\draw[->] (2a) to [] (5);
		\draw [->] (4) to [out=35,in=110, looseness=1] (4a);
		\draw [->] (4a) to [out=-110,in=-35,looseness=1] (4);
		\draw [->] (5) -- (5a);
		\draw [->] (5a) -- node[below, pos=0.3, xshift=0mm]{\scriptsize $0.8$} (2);
		\draw [->] (5a) -- node[right, xshift=0mm]{\scriptsize $0.2$} (4);
		\draw [->] (3) -- (3a);
		\draw [->] (3a) -- (6);
		\draw[->] (2) -- (2b);
		\draw[->] (2b) -- (6);
		\draw [->] (6) to [out=125,in=200, looseness=1] (6a);
		\draw [->] (6a) to [out=-20,in=55,looseness=1] (6);
		\end{tikzpicture}
	}
		\caption{MDP $\mdp$ with $\goalstates_\varocircle {=} \set{\state_4,\state_6}$ and $\goalstates_\square {=} \set{\state_3}$}
		\label{fig:objectives:mdp}
		\end{subfigure}
		\begin{subfigure}[b]{0.34\linewidth}
			\centering
		\newcommand{\paretoplotscale}{0.24}
		\newcommand{\paretoplotwidth}{0.243\textwidth}
		\newcommand{\paretoplotxmax}{10.6}
		\newcommand{\paretoplotymax}{10.6}
		\newcommand{\paretoplotmin}{-1}
		\centering
		\begin{tikzpicture}[scale=\paretoplotscale,yscale=0.88]
	\path[use as bounding box] (-1,-1.5) rectangle (10,10);
				\path[genach] (-0.5,-0.5) -- (10,-0.5) -- (10,8) -- (0, 10) (-0.5,10) -- cycle;
				\path[ach] (-0.5,-0.5) rectangle (7,7);
				\path[ach] (-0.5,-0.5) rectangle (10,0);
				\path[ach] (-0.5,-0.5) rectangle (0,10);
				\path[ach] (-0.1,-0.1) -- (10,-0.1) -- (10, 0.1) -- (0.1,0.1) -- cycle;
				\path[ach] (-0.1,-0.1) -- (-0.1,10) -- (0.1, 10) -- (0.1,0.1) -- cycle;
				\node[point,label=-90:{\scriptsize$\tuple{0.7,0.7}$}] at (7,7) (ps1) {};
				\node[point,label=-45:{\scriptsize$\tuple{0,1}$}] at (0,10) (ps2) {};
				\node[point,label=90:{\scriptsize$\tuple{1,0}$}] at (10,0) (ps3) {};
				\node[,label=0:{\scriptsize$\tuple{1,0.8}$}] at (10,8) (p4) {};
				\path[->]
				(\paretoplotmin,0) edge (\paretoplotxmax,0)
				(0,\paretoplotmin) edge (0,\paretoplotymax)
				;
				\node[empty] at (5, -1.2) {\scriptsize $\reachobjop_\geq(\eventually \goalstates_\varocircle)$};
				\node[empty,rotate=90] at (-1.2, 5) {\scriptsize $\reachobjop_\geq(\eventually \goalstates_\square)$};
				\path[genborder] (0,10) edge (10,8);
		\end{tikzpicture}
			\caption{\mbox{$\ach{\mdp}{\multiobjquery}$ and $\psach{\mdp}{\multiobjquery}$}}
		\label{fig:objectives:plot}
	\end{subfigure}
		\caption{An MDP and a plot of the pure stationary and general Pareto fronts.}
		\label{fig:objectives}
\end{figure}
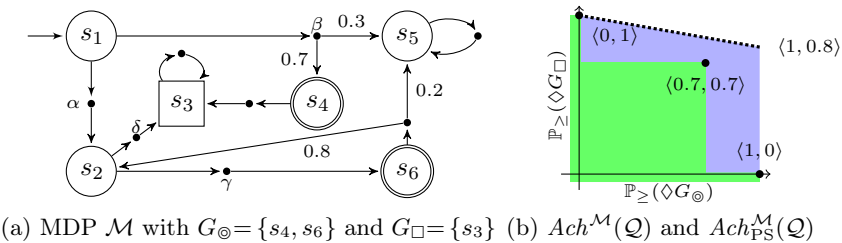 

\begin{example}
	Let $\mdp$ be the MDP in \Cref{fig:objectives:mdp} and $\multiobjquery = \tuple{\reachobjop_\geq(\eventually \goalstates_\varocircle), \reachobjop_\geq(\eventually \goalstates_\square)}$.
	A pure stationary strategy
	choosing $\beta$ at $\state_1$ reaches both, $\state_4 \in \goalstates_\varocircle$ and $\state_3 \in \goalstates_\square$ with probability $0.7$ and thus achieves $\tuple{0.7, 0.7}$. Similarly, $\tuple{0,1}$ and $\tuple{1,0}$ are achievable by a pure stationary strategy.
	Point $\tuple{1,0.8}$ is achievable by a non-stationary pure strategy that chooses $\alpha$ at $\state_1$, $\gamma$ at the first visit of $\state_2$, and $\delta$ in all other cases.
	Changing this strategy by picking $\gamma$ only with probability 0.5 achieves $\tuple{0.5,0.9}$.
	\Cref{fig:objectives:plot} illustrates $\pspareto{\mdp}{\multiobjquery}$ (dots), $\psach{\mdp}{\multiobjquery}$ (green area), $\pareto{\mdp}{\multiobjquery}$ (dotted line), and $\ach{\mdp}{\multiobjquery}$ (blue and green area).
\end{example}
\section{Deciding Achievability}
\label{sec:achievability}
The achievability problem asks whether a given point is achievable.
\begin{nproblem}[framed]{General Multi-objective Achievability Problem (\gma)}
	Input: & MDP $\mdp$, $\numobj$-dimensional multi-objective query $\multiobjquery$, point $\point \in \pointdomain$\\
	Output: & Yes iff $\point \in \ach{\mdp}{\multiobjquery}$ holds
\end{nproblem}
\noindent
For \gma, the point can be achieved by a general strategy that can potentially make use of memory and randomization.
As discussed earlier, this class of strategies is not suitable for various applications.
In this work, we focus on a variant of the achievability problem that only considers pure stationary strategies.
\Cref{sec:boundedmem} also addresses pure strategies that can store more information from the history, \eg whether a goal state set has been reached already.

\begin{nproblem}[framed]{Pure Stationary Multi-objective Achievability Problem (\psma)}
	Input: & MDP $\mdp$, $\numobj$-dimensional multi-objective query $\multiobjquery$, point $\point \in \pointdomain$\\
	Output: & Yes iff $\point \in \psach{\mdp}{\multiobjquery}$ holds
\end{nproblem}

\subsection{Complexity Results}

\gma is PSPACE hard (already with only reachability objectives)~\cite{RRS17} and  solvable within exponential runtime~\cite{EKVY07,FKNPQ11}.
To the best of our knowledge, a PSPACE upper bound on the complexity of \gma is unknown. This complexity is rooted in the dimension $\numobj$ of the query $\multiobjquery$:
for fixed $\numobj$, the algorithms of~\cite{EKVY07,FKNPQ11} have polynomial runtime.
In contrast, $\psma$ is NP-complete, even if restricted to  2~objectives.

\begin{lemma}\label{psma-reachability}
	\psma with only reachability objectives is NP-hard.
\end{lemma}

\begin{proof}
	The result follows by a reduction from the subset sum problem.
	Given $n \in \mathbb{N}$, $\vec{a} \in \mathbb{N}^n$ and $z \in \mathbb{N}$, the subset sum problem is to decide the existence of $\vec{v} \in \{0, 1\}^n$ such that $\vec{v} \cdot \vec{a} = z$.
	This problem is NP-complete \cite{GJ79}.
	For a given instance of the subset sum problem, we construct the MDP $\mdp^\star = \mdptuple$ with state space $\states = \set{\sinit, \state_1, \dots, \state_n, g_1, g_2}$, actions $\actions = \set{\action, Y, N}$, and for all $i \in \set{1, \dots, n}$, 
	$\transitions(\sinit, \action, \state_i) = \frac{\tupleaccess{\vec{a}}{i}}{\vec{1 \cdot \vec{a}}}$ and
	$\transitions(\state_i, Y, g_1) = \transitions(\state_i, N, g_2) = 1$.
	States $g_1$ and $g_2$ are made absorbing, \ie $\transitions(g_1, \action, g_1) = \transitions(g_2, \action, g_2) =  1$.
	
	We claim that the \psma problem for $\mdp^\star$, $\multiobjquery = \tuple{ \reachobjop_\geq(\eventually \set{g_1}),  \reachobjop_\geq(\eventually \set{g_2})}$, and $\point = \left(\frac{z}{\vec{1} \cdot \vec{a}}, 1- \frac{z}{\vec{1} \cdot \vec{a}} \right)$ answers ``yes'' iff there is a vector $\vec{v}$ satisfying the subset sum problem for $n$, $\vec{a}$ and $z$.
	Consider the bijection $f \colon \psscheds{\mdp^\star} \to  \{0, 1\}^n$ with $\tupleaccess{f(\sched)}{i} = \iverson{\sched(\state_i) {=} Y}$ for all $\sched \in \psscheds{\mdp^\star}$ and $i\in\set{1,\dots,n}$.
	We get $\probmeasure{\mdp^\star}{\sched}(\eventually \set{g_1}) = \sum_{i=1}^n   \frac{\tupleaccess{\vec{a}}{i}}{\vec{1 \cdot \vec{a}}} \iverson{\sched(\state_i) {=} Y} = \frac{f(\sched) \cdot \vec{a}}{\vec{1} \cdot \vec{a}}$.
		Moreover, $\probmeasure{\mdp^\star}{\sched}(\eventually \set{g_2}) = 1-\probmeasure{\mdp^\star}{\sched}(\eventually \set{g_1}) = 1 - \frac{f(\sched) \cdot \vec{a}}{\vec{1} \cdot \vec{a}}$.
	It follows that $\sched$ achieves $\point$ iff $f(\sched)$ is a solution to the instance of the subset sum problem.
	Our construction is inspired by similar ideas from~\cite{CMH06,RRS17}.
\end{proof}

\begin{lemma}[\cite{CMH06}]\label{psma-reward}
	\psma with only total reward objectives is NP-hard.
\end{lemma}

\begin{theorem}
	\psma is NP-complete.
\end{theorem}
\begin{proof}
	Containment follows by guessing a pure stationary strategy and evaluating it on the individual objectives. This can be done in polynomial time~\cite{BK08}.
	Hardness follows by either Lemma~\ref{psma-reachability} or~\ref{psma-reward}.
\end{proof}

Proofs of Lemmas~\ref{psma-reachability} and~\ref{psma-reward} only consider 2-dimensional multi-objective queries. Hence, in contrast to \gma, the hardness of \psma is not due to the size of the query.
\begin{corollary}
	\psma with only two objectives is NP-complete.
\end{corollary}


\subsection{A Mixed Integer Linear Programming Approach}\label{sec:milp}

An MDP $\mdp = \mdptuple$ has exactly $|\psscheds{\mdp}| = \prod_{\state \in \states} |\act{\state}|$ many pure stationary strategies.
A simple algorithm for \psma enumerates all $\sched \in \psscheds{\mdp}$ and checks whether $\tuple{\mdp, \sched, \point} \models \multiobjquery$ holds.
In practice, however, such a brute-force approach is not feasible.
For the MDPs that we consider in our experiments in Sect.~\ref{sec:evaluation}, the number of pure stationary strategies often exceeds $10^{10\,000}$.
Instead, our approach is to encode an instance for \psma as an \milp problem.

\begin{nproblem}[framed]{Mixed Integer Linear Programming Problem (\milp)}
	Input: & $\ell,m,n \in \nn$, $\vec{A} \in \qq^{n \times (\ell + m)}$, $\vec{b} \in \qq^{n}$, $\vec{c} \in \qq^{\ell + m}$\\
	Output: & $\begin{cases} \vec{x} \in \argmax_{\vec{x} \in \mathcal{X}} \vec{c}^T \vec{x} & \text{if } \mathcal{X} \neq \emptyset  \\ \mathtt{infeasible} & \text{if } \mathcal{X} = \emptyset \end{cases}$
	with $\mathcal{X} = \{ \vec{x} \in \zz^\ell \times \rr^m \mid \vec{A}\vec{x} \le \vec{b} \}$
\end{nproblem}
For an \milp instance as above, each of the $n$ rows of the inequation system $\vec{A}\vec{x} \le \vec{b}$ represent a \emph{constraint} that is linear over the $\ell$ integral and $m$ real-valued \emph{variables} given by $\vec{x}$.
We call the constraints \emph{feasible} if there is a solution to the inequation system.
The task is to decide whether the constraints are feasible and if so, find a solution that maximizes a linear optimization function $\vec{c}^T \vec{x}$.
The optimization function can be omitted if we are only interested in feasibility.
\milp is NP-complete~\cite{PDM17}.
However, tools such as \tool{Gurobi}~\cite{gurobi} and \tool{SCIP}~\cite{scip} implement practically efficient algorithms that  can solve large instances.

For the rest of this section, let $\mdp = \mdptuple$, $\multiobjquery = \tuple{\obj_1, \dots, \obj_\numobj}$ with $\objj = \rewobjop_{\objrelj}(\rewstructj \eventually \goalstatesj)$, and $\point \in \pointdomain$ be an instance for \psma.
%
We provide a translation of the \psma instance to an instance for \milp that has a feasible solution iff $\point \in \psach{\mdp}{\multiobjquery}$.
The \milp encoding considers integer variables to encode a pure stationary strategy $\sched \in \psscheds{\mdp}$.
The other variables and constraints encode the expected reward for each objective on the induced MC $\mc{\mdp}{\sched}$.


%
\subsection{Unichain MDP and Finite Rewards}\label{sec:achievability:unichain}

\begin{restr}[Unichain MDP]\label{as:unichain}
	MDP $\mdp$ has exactly one end component.
\end{restr}
\begin{restr}[Reward Finiteness]\label{as:rewardfiniteness}
	$ \expval{\mdpwithinitstate{\mdp}{\state}}{\sched}(\rewstructj \eventually \goalstatesj) < \infty$ holds for each objective $\objj = \rewobjop_{\objrelj}(\rewstructj \eventually \goalstatesj)$, state $\state$, and pure stationary strategy $\sched$.
\end{restr}
For simplicity, we first explain our encoding for unichain MDP with finite reward.
\Cref{sec:achievability:multichain} lifts \Cref{as:unichain} and
\Cref{sec:achievability:infinite} lifts \Cref{as:rewardfiniteness} with more details given in \reftr{B}.
\Cref{sec:achievability:altencoding} presents an alternative to the encoding of this section, which is smaller but restricted to \emph{total} reward objectives.

\Cref{lp:unichainfinite} shows the \milp encoding in case \Cref{as:unichain,as:rewardfiniteness} hold.
We assume $\fa{\objindex}\tupleaccess{\point}{\objindex} \neq \infty$ for the point $\point$ since (i) $\expval{\mdp}{\sched}(\rewstructj \eventually \goalstatesj) \le \infty$ holds trivially and (ii) $\expval{\mdp}{\sched}(\rewstructj \eventually \goalstatesj) \ge \infty$ will never hold due to \Cref{as:rewardfiniteness}.
For $\objindex \in \set{1, \dots, \numobj}$, 
let
$
\stateszeroj =  \{\state \in \states \mid \fa{\sched \in \scheds{\mdp}}\expval{\mdp}{\sched}(\rewstructj \eventually \goalstatesj) = 0\}
$
 and
$\statesmaybej = \{ \state \in \states \setminus \stateszeroj \mid \allowbreak\state \text{ can be reached from } \sinit \text{ without visiting a state in } \stateszeroj\}$.
These sets can be obtained a priori by analyzing the graph structure of $\mdp$~\cite{BK08}.
Moreover, we consider upper bounds $\uppervalueboundsj \in \qq$ for the expected reward at state $\state \in \statesmaybej$ such that $\uppervalueboundsj \ge \max_{\sched \in \scheds{\mdp}}  \expval{\mdpwithinitstate{\mdp}{\state}}{\sched}(\rewstructj \eventually \goalstatesj)$.
We compute such upper bounds using single-objective model checking techniques~\cite{BK08,BKLPW17}.
The $\milp$ encoding applies the characterization of expected rewards for MCs as a \emph{linear equation system}~\cite{BK08}.
\begin{lemma}
	\label{lem:mceqsys}
	For every $\sched \in \psscheds{\mdp}$, 
	the following equation system has a unique solution $\lpsol \colon \set{\lpvalvar{\state}{} \mid \state \in \states} \to \rr^{|\states|}$ satisfying $\lpsolof{\lpvalvar{\state}{}} = \expval{\mdpwithinitstate{\mdp}{\state}}{\sched}(\rewstructj\eventually\goalstatesj)$:
	\begin{align*}
	\fa{\state \in \stateszeroj} \lpvalvar{\state}{}=0
	\qquad\quad 
	\fa{\state \in \statesmaybej} \lpvalvar{\state}{}=\sum_{\state'\in\states} \transitions(\state,\sched(\state),\state') \cdot \big(\lpvalvar{\state'}{} + \rewstruct(\state,\sched(\state),\state') \big)
	\end{align*}
\end{lemma}
\begin{proof}
	Since $\mdp$ is unichain and we do not collect infinite reward, the only EC of $\mdp$ (\ie the only BSCC of $\mc{\mdp}{\sched}$ for any $\sched$) either contains a goal state or only contains transitions with zero reward.
	It follows that $\fa{\sched \in \psscheds{\mdp}}\probmeasure{\mdp}{\sched}(\eventually \stateszeroj) = 1$.
	\Cref{lem:mceqsys} follows by standard arguments for MCs with rewards \cite[Section 10.5.1]{BK08}. 
\end{proof}

\begin{figure}[t]
	\begin{lp}
	\lpline{}{\fa{\state \in \states}}{}{}{Select an action at each state}{}
	\lpline{"}{\fa{\action \in \act{\state}}}{\lpactionvarsa}{\in \set{0,1}}{}{lp:action:start}
	\lpline{"}{}{\sum_{\action \in \act{\state}}\lpactionvarsa}{=1}{}{lp:action:end}
	\lpline{}{\fa{\objindex \in \set{1, \dots, \numobj}}}{}{}{Compute expected reward values}{}
	\lpline{"}{\fa{\state \in \stateszeroj}}{\lpvalvarsj}{= 0}{}{lp:exprew:start}
	\lpline{"}{\text{If $\obj_\objindex$ is maximising, ${\pmobj} = {+}$ and $\minobj = 0$. Otherwise, ${\pmobj} = {-}$ and $\minobj = 1$.}}{}{}{}{}
	\lpline{"}{\fa{\state \in \statesmaybej}}{\pmobj\, \lpvalvarsj}{\in [0,\uppervalueboundsj]}{}{}
	\lpline{""}{\fa{\action \in \act{\state}}}{\pmobj\, \lpvalactionvarsaj}{\in [0,\uppervalueboundsj]}{}{}
	\lpline{"""}{}{\lpvalactionvarsaj}{\le \sum_{\state' \in \states} \transitions(\state,\action,\state') \cdot   \left(\lpvalvar{\state'}{\objindex} \pmobj \rewstructj(\state,\action,\state')\right)}{}{lp:exprew:sum}
	\lpline{"""}{}{\lpvalactionvarsaj}{\le \uppervalueboundsj \cdot \left(\lpactionvarsa - \minobj\right)}{}{lp:exprew:disabled}
	\lpline{""}{}{\lpvalvarsj}{\le \sum_{\action \in \act{\state}} \lpvalactionvarsaj + \minobj \cdot (|\act{\state}| - 1) \cdot \uppervalueboundsj}{}{lp:exprew:end}
	\lpline{"}{}{\pmobj\, \lpvalvar{\sinit}{\objindex}}{\objrelj \tupleaccess{\point}{\objindex}}{Assert value at initial state}{lp:exprew:point}
	\end{lp}
	\caption{\milp encoding for unichain MDP and finite rewards.}
	\label{lp:unichainfinite}
\end{figure}
We discuss the intuition of each constraint in \Cref{lp:unichainfinite}.
Let $\lpsol \colon \var \to \rr$ be an assignment of the occurring variables $\var$ to values.
$\lpsol$ is a solution of the constraints if all (in)equations are satisfied upon replacing all variables $v$ by~$\lpsolof{v}$.

Lines~\ref{lp:action:start} and~\ref{lp:action:end} encode a strategy $\sched \in \psscheds{\mdp}$ by considering a binary variable $\lpactionvarsa$ for each state $\state$ and enabled action $\action$ such that $\sched(\state)(\action) = 1$ iff $\lpsolof{\lpactionvarsa} = 1$ for a solution $\lpsol$. Due to Line~\ref{lp:action:end}, exactly one action has to be chosen at each state.

Lines~\ref{lp:exprew:start} to~\ref{lp:exprew:end} encode for each objective $\objj$ the expected rewards obtained for the encoded strategy $\sched$.
For every $\state \in \states$, the variable $\lpvalvarsj$ represents a (lower or upper) bound on the expected reward at $\state$.
Line~\ref{lp:exprew:start} sets this value for all $\state \in \stateszeroj$, reflecting the analogous case from \Cref{lem:mceqsys}.
For $\state \in \statesmaybej$, we distinguish maximizing (${\objrelj} = {\ge}$) and minimizing (${\objrelj} = {\le}$) objectives $\objj$.

 For maximizing $\objj$, we have $\lpsolof{\lpvalvarsj} \le  \expval{\mdpwithinitstate{\mdp}{\state}}{\sched}(\rewstructj \eventually \goalstatesj)$ for every solution $\lpsol$.
	This is achieved by considering a variable $\lpvalactionvarsaj$ for each enabled action $\action \in \act{\state}$.
	In Line~\ref{lp:exprew:sum}, we use the equation system characterization from~\Cref{lem:mceqsys} to assert that the value of $\lpvalactionvarsaj$ can not be greater than the expected reward at $\state$, given that  the encoded strategy $\sched$ selects $\action$.
	If $\sched$ does not select $\action$ (\ie $\lpsolof{\lpactionvarsa} = 0$), Line~\ref{lp:exprew:disabled} implies $\lpsolof{\lpvalactionvarsaj} = 0$.
	Otherwise, this constraint has no effect.
	Line~\ref{lp:exprew:end} ensures that every solution satisfies $\lpsolof{\lpvalvarsj} \le \lpsolof{\lpvalactionvarsaj} \le  \expval{\mdpwithinitstate{\mdp}{\state}}{\sched}(\rewstructj \eventually \goalstatesj)$ for $\action$ with $\lpsolof{\lpactionvarsa} = 1$.

	For minimizing $\objj$, we have $-\lpsolof{\lpvalvarsj} \ge  \expval{\mdpwithinitstate{\mdp}{\state}}{\sched}(\rewstructj \eventually \goalstatesj)$ for every solution $\lpsol$, \ie we consider the negated reward values.
	The encoding is as for maximizing objectives. However, Line~\ref{lp:exprew:disabled} yields $\lpsolof{\lpvalactionvarsaj} = - \uppervalueboundsj$ if $\action$ is not selected.
	Thus, in Line~\ref{lp:exprew:end} we add $\uppervalueboundsj$ for each of the $(|\act{\state}| - 1)$ non-selected actions.

Line~\ref{lp:exprew:point} and our observations above yield $\expval{\mdpwithinitstate{\mdp}{\state}}{\sched}(\rewstructj \eventually \goalstatesj) \ge \lpsolof{\lpvalvar{\sinit}{\objindex}} \ge \tupleaccess{\point}{\objindex}$
for maximizing and
$\expval{\mdpwithinitstate{\mdp}{\state}}{\sched}(\rewstructj \eventually \goalstatesj) \le -\lpsolof{\lpvalvar{\sinit}{\objindex}} \le \tupleaccess{\point}{\objindex}$
for minimizing objectives.
Therefore, $\point$ is achievable if a solution $\lpsol$ exists.
On the other hand, if $\point$ is achievable by some $\sched \in \psscheds{\mdp}$, the solution $\lpsol$ exists with $\lpsolof{\lpactionvarsa} = \sched(\state)(\action)$, $\lpsolof{\lpvalvarsj} = \lpsolof{\lpvalactionvarsaj} = \pmobj \expval{\mdpwithinitstate{\mdp}{\state}}{\sched}(\rewstructj \eventually \goalstatesj)$ if $\action = \sched(\state)$, and $\lpsolof{v} = 0$ for other~$v \in \var$.

\begin{theorem}\label{thm:unichainfinite}
	For unichain $\mdp$ and finite rewards, the constraints in~\Cref{lp:unichainfinite} are feasible iff
	$\point \in \psach{\mdp}{\multiobjquery}$.
\end{theorem}

\begin{proposition}\label{pro:unichainfinitesize}
	The \milp encoding above considers $\mathcal{O}(|\states|\cdot |\actions| \cdot \numobj)$ variables.
\end{proposition}

\subsection{Alternative Encoding for Total Rewards} \label{sec:achievability:altencoding}
We now consider \psma instances where all objectives $\objj = \rewobjop_{\objrelj}(\rewstructj \eventually \goalstatesj)$ are expected \emph{total} reward objectives, \ie $\goalstatesj = \emptyset$.
For such instances, we can employ an encoding from~\cite{FKNPQ11} (restated in~\Cref{lem:flow}) for \gma.
In fact, we can often translate reachability reward objectives to total reward objectives, \eg if the set of goal states can not be left or if all objectives consider the same goal states.

\begin{lemma}[\cite{FKNPQ11}]\label{lem:flow}
	For $\stateszero \subseteq \states$, let $\lpsol \colon \var \to \rewarddomain$ be an assignment of variables $\var = \set{\lpflowvarsa \mid \state \in \states \setminus \stateszero, \action \in \act{\state}}$ and let
	$\sched_\lpsol$ be a stationary strategy satisfying $\sched_\lpsol(\state)(\action) = \lpsolof{\lpflowvarsa} / \sum_{\beta \in \act{\state}} \lpsolof{\lpflowvar{\state}{\beta}}$ for all $\state \in \states \setminus \stateszero$ and $\action \in \act{\state}$ for which the denominator is non-zero.
	Then, $\lpsol$ is a solution to the equation system
	\begin{align*}
	\fa{\state \in \states \setminus \stateszero}
	\sum_{\action \in \act{\state}} \lpflowvarsa
	~&=~
	\iverson{\state = \sinit} \,\,\,+ \!\!\!\sum_{\tuple{\state',\action'} \in \predecessors{\state}} \transitions(\state',\action',\state) \cdot \lpflowvar{\state'}{\action'} 
	\\
	1
	~&=~
	\sum_{\lpflowvarsa \in \var} \lpflowvarsa \cdot  \sum_{\state' \in \stateszero} \transitions(\state,\action,\state') 
	\end{align*}
	%
	iff $\probmeasure{\mdp}{\sched_\lpsol}(\eventually \stateszero) = 1$ and
	$\fa{\lpflowvarsa \in \var} \lpsolof{\lpflowvarsa} = \expval{\mdp}{\sched_\lpsol}(\rewstruct_{\state,\action} \eventually \stateszero)$
	with reward structure $\rewstruct_{\state,\action}$ given by $\rewstruct_{\state,\action}(\hat{\state},\hat{\action}, s') = \iverson{\hat{\state} = \state \text{ and } \hat{\action} = \action}$. 
\end{lemma}
In \cite{FKNPQ11}, the lemma is applied to decide achievability of multiple total reward objectives under strategies that are stationary, but not necessarily pure.
Intuitively, $\expval{\mdp}{\sched_\lpsol}(\rewstruct_{\state,\action} \eventually \stateszero)$ coincides with the expected number of times action $\action$ is taken at state $\state$ until $\stateszero$ is reached.
Since this value can be infinite if $\probmeasure{\mdp}{\sched_\lpsol}(\eventually \stateszero) < 1$, a solution $\lpsol$ can only exist if it induces a strategy that almost surely reaches $\stateszero$.
\begin{figure}[t]
	\begin{lp}
	\lpline{}{\fa{\state \in \statesmaybe,\action \in \act{\state}}\enlargecolumn}{\lpflowvarsa}{\in [0, \maxflow{\state} \cdot \lpactionvarsa]}{}{lp:alternative:start}
	\lpline{"}{~~~~~~~~~~~~~~~~~~~~~}{ \sum_{\action \in \act{\state}} \lpflowvarsa}{=\iverson{\state = \sinit} + \sum_{\tuple{\state',\action'} \in \predecessors{\state}} \transitions(\state',\action',\state) \cdot \lpflowvar{\state'}{\action'}}{}{}
	\lpline{}{}{1}{=\sum_{\state \in \statesmaybe} \sum_{\action \in \act{\state}} \lpflowvarsa \cdot \sum_{\state' \in \stateszero} \transitions(\state,\action,\state')}{}{lp:alternative:outflow}
	\lpline{"}{\fa{\objindex \in \set{1, \dots, \numobj}}\enlargecolumn}{\lpvalvar{\sinit}{\objindex}}{= \sum_{\state \in \statesmaybe} \sum_{\action \in \act{\state}} \lpflowvarsa \cdot \sum_{\state' \in \states} \left(\transitions(\state,\action,\state') \cdot  \rewstructj(\state, \action, \state')\right) }{}{lp:alternative:value}
	\lpline{"}{}{\lpvalvar{\sinit}{\objindex}}{\objrelj \tupleaccess{\point}{\objindex}}{}{lp:alternative:end}
	\end{lp}
	\caption{\milp encoding for total reward objectives.}
	\label{lp:alternative}
\end{figure}

The encoding for unichain MDP with finite rewards and total reward objectives is shown in \Cref{lp:alternative}, where $\stateszero = \bigcap_{j} \stateszeroj$ and $\statesmaybe = \states \setminus \stateszero$.
We consider the constraints in conjunction with Lines~\ref{lp:action:start} and~\ref{lp:action:end} from \Cref{lp:unichainfinite}.
Let $\lpsol$ be a solution and let $\sched$ be the strategy encoded by such a solution, \ie $\sched(\state)(\action) = \lpsolof{\lpactionvarsa}$.

Lines~\ref{lp:alternative:start} to~\ref{lp:alternative:outflow} reflect the equations of~\Cref{lem:flow}.
Since $\mdp$ is unichain and we assume finite rewards, there is just one end component in which no reward can be collected. Hence, $\stateszero$ is almost surely reached.
Line~\ref{lp:alternative:start} ensures that the strategy in~\Cref{lem:flow} coincides with the encoded pure strategy $\sched$.
We write $\maxflow{\state}$ for an upper bound of the value a solution can possibly assign to $\lpflowvarsa$, \ie
$
\fa{\sched \in \psscheds{\mdp}}
\maxflow{\state} \ge 
\expval{\mdp}{\sched}(\rewstruct_{\state,\action} \eventually  \stateszero)
$.
Such an upper bound can be computed based on ideas of~\cite{BKLPW17}. More details are given in~\reftr{A}.

%
%

With \Cref{lem:flow} we get that $\lpsolof{\lpflowvar{\state}{\sched(\state)}}$ is the expected number of times state $s$ is visited under strategy $\sched$.
Therefore, 
in Line~\ref{lp:alternative:value} we sum up for each state $\state \in \statesmaybe$ the expected amount of reward collected at $\state$.
This yields $\lpsolof{\lpvalvar{\sinit}{\objindex}} = \expval{\mdp}{\sched}(\rewstructj \eventually \goalstatesj)$.
Finally, Line~\ref{lp:alternative:end} asserts that the resulting values exceed the thresholds given by~$\point$.
\begin{theorem}\label{thm:alternative}
	For unichain $\mdp$, finite rewards, and total reward objectives, the constraints in \Cref{lp:alternative} and Lines~\ref{lp:action:start} and~\ref{lp:action:end} of \Cref{lp:unichainfinite} are feasible iff
	$\point \in \psach{\mdp}{\multiobjquery}$.
\end{theorem}

\begin{proposition}\label{pro:unichainfinitetotalsize}
	The \milp encoding above considers $\mathcal{O}(|\states|\cdot |\actions| + \numobj)$ variables.
\end{proposition}

The encoding for total reward objectives considers fewer variables compared to the encoding of \Cref{sec:achievability:unichain} (cf. \Cref{pro:unichainfinitesize}).
In practice, this often leads to faster solving times as we will see in \Cref{sec:evaluation}.

\subsection{Extension to Multichain MDP}\label{sec:achievability:multichain}
We now lift the restriction to unichain MDP, \ie we consider multichain MDP with finite rewards.
We focus on the encoding of \Cref{sec:achievability:unichain}.
Details for the approach of \Cref{sec:achievability:altencoding} are in \reftr{C}.
The key challenge is that the equation system in \Cref{lem:mceqsys} does not yield a unique solution for multichain MDP.
\begin{example}\label{ex:multichain}
	For the multichain MDP in \Cref{fig:multichain} with $\goalstates = \set{\state_1}$ we have $\stateszero = \set{\state_1}$ and $\statesmaybe = \set{\state_0}$ (the superscript~$\objindex$ is omitted as there is only one objective).
	For $\sched$ with $\sched(\state_0) = \alpha$ we get $\expval{\mdp}{\sched}(\rewstruct \eventually \goalstates) = 0$, but every $\lpsol \colon \set{\lpvalvar{\state_0}{},\lpvalvar{\state_1}{}} \to \rr \times \set{0} $  is a solution for the equation system in \Cref{lem:mceqsys}.
\end{example}
For multichain MDP it can be the case that for some strategy $\sched$ the set $\stateszeroj$ is not reached with probability 1, \ie there is a positive probability to stay in the set $\statesmaybej$ forever.
For the induced Markov chain $\mc{\mdp}{\sched}$, this means that there is a reachable BSCC consisting only of states in $\statesmaybej$.
Since BSCCs of $\mc{\mdp}{\sched}$ coincide with end components of $\mdp$, we need to inspect the ECs of $\mdp$ that only consist of $\statesmaybej$-states.
These ECs correspond to the ECs of the sub-MDP $\submdp{\mdp}{\stateactionpairset_?^\objindex }$, where 
$\stateactionpairset_?^\objindex$ is the largest subset of $\statesmaybej \times \actions$ that is closed for $\mdp$.
For each $\stateactionpairset \in \mecs{\submdp{\mdp}{\stateactionpairset_?^\objindex}}$, we need to detect whether the encoded strategy induces a BSCC $\stateactionpairset' \subseteq \stateactionpairset$.

\begin{figure}[t]
	\begin{lp}
	\lpline{}{\fa{\objindex \in \set{1, \dots, \numobj}, \stateactionpairset \in \mecs{\submdp{\mdp}{\stateactionpairset_?^\objindex}}}}{}{}{Detect states with zero reward}{}
	\lpline{"}{\fa{\state \in \submdp{\states}{\stateactionpairset}}}{\pmobj\lpvalvarsj}{\le \uppervalueboundsj \cdot (1-\lpecvarsj)}{}{lp:ec:val}
	\lpline{"}{\fa{\tuple{\state,\action} \in \stateactionpairset}\enlargecolumn}{\lpecactionvarsaj}{\in \set{0,\lpactionvarsa}}{}{lp:ec:graph:start}
	\lpline{""}{\fa{\state' \in \successors{\state}{\action}}\enlargecolumn}{\lpecactionvarsaj}{\le \lpecvar{\state'}{\objindex}}{}{lp:ec:graph:succ}
	\lpline{"}{\raisebox{0pt}{\fa{\state \in \submdp{\states}{\stateactionpairset}}}}{\lpecvarsj}{= \sum_{\action \in \act{\state}} \iverson{\tuple{\state,\action} \in \stateactionpairset} \cdot \lpecactionvarsaj}{}{lp:ec:graph:end}
	

	\lpline{""}{\fa{\action \in \act{\state}}\enlargecolumn}{\lpecflowvarsaj}{\in [0, \maxflow{s} \cdot \lpactionvarsa]}{}{lp:ec:flow:start}
	\lpline{""}{}{\lpecflowvar{\state}{\bot}{\objindex}}{\in [0, \maxflow{s} \cdot \lpecvarsj]}{}{lp:ec:flow:outvar}
	\lpline{""}{}{\lpecflowvar{\state}{\actbot}{\objindex} + \sum_{\action \in \act{\state}} \lpecflowvarsaj}{=\frac{1}{|\submdp{\states}{\stateactionpairset}|} + \sum_{\tuple{\state', \action'} \in \predecessors{\state} \cap \stateactionpairset} \transitions(\state',\action',\state) \cdot \lpecflowvar{\state'}{\action'}{\objindex}}{}{lp:ec:flow:inout}
	\lpline{"}{}{1}{= \sum_{\state \in \submdp{\states}{\stateactionpairset}} \Big( \lpecflowvar{\state}{\actbot}{\objj} + \sum_{\action \in \act{\state}} \iverson{\tuple{\state,\action} \notin \stateactionpairset} \cdot \lpecflowvarsaj\Big)}{}{lp:ec:flow:end}
	\end{lp}
	\caption{\milp encoding for detection of end components.}
	\label{lp:endcomponents}
\end{figure}

To cope with multiple ECs, we consider the constraints from \Cref{lp:unichainfinite} in conjunction with the constraints from \Cref{lp:endcomponents}.
Let $\lpsol$ be a solution to these constraints
and let $\sched$ be the encoded strategy $\sched$ with $\sched(\state)(\action) = \lpsolof{\lpactionvarsa}$.
For each objective $\objj$ and state $\state$, a binary variable $\lpecvarsj$ is set to $1$ if $\state$ lies on a BSCC of the induced MC $\mc{\mdp}{\sched}$.
We only need to consider states $\state \in \submdp{\states}{\stateactionpairset}$ for $\stateactionpairset \in \mecs{\submdp{\mdp}{\stateactionpairset_?^j}}$.

Line~\ref{lp:ec:val} ensures that the value of $\lpvalvarsj$ is set to 0 if $\state$ lies on a BSCC of $\mc{\mdp}{\sched}$.
Lines~\ref{lp:ec:graph:start} to~\ref{lp:ec:graph:end} introduce binary variables $\lpecactionvarsaj$ for each state-action pair in the EC such that any solution $\lpsol$ satisfies $\lpsolof{\lpecactionvarsaj} = 1$ iff $\lpsolof{\lpecvarsj} = \lpsolof{\lpactionvarsa} = 1$.
Line~\ref{lp:ec:graph:succ} yields that $\lpsolof{\lpecactionvarsaj} = 1$ implies $\lpsolof{\lpecvar{\state'}{\objindex}} = 1$ for all successors $\state'$ of $\state$ and the selected action $\action$.
Hence, for all $\state$ with $\lpsolof{\lpecvarsj} = 1$ and for all $\state'$ reachable from $\state$ in $\mc{\mdp}{\sched}$, we have $\lpsolof{\lpecvar{\state'}{\objindex}} = 1$ and $\tuple{\state', \sched(\state')} \in \stateactionpairset$.
Therefore, we can only set $\lpecvarsj$ to 1 if there is a BSCC $\stateactionpairset' \subseteq \stateactionpairset$ that either contains $s$ or that is almost surely reached from $\state$ without leaving $\stateactionpairset$.
As finite rewards are assumed, $\stateactionpairset$ can not contain a transition with positive reward, yielding  $\expval{\mdp}{\sched}(\rewstructj \eventually\goalstatesj) = 0$ if $\lpsolof{\lpecvarsj} = 1$.

An assignment that sets all variables $\lpecvarsj$ and $\lpecactionvarsaj$ to 0 trivially satisfies the constraints in Lines~\ref{lp:ec:val} to~\ref{lp:ec:graph:end}.
In Lines~\ref{lp:ec:flow:start} to~\ref{lp:ec:flow:end} we therefore ensure that if a BSCC $\stateactionpairset' \subseteq \stateactionpairset$ exists in $\mc{\mdp}{\sched}$, $\lpsolof{\lpecvarsj} = 1$ holds for at least one $\state \in \submdp{\states}{\stateactionpairset'}$.
The idea is based on the observation that if a BSCC $\stateactionpairset' \subseteq \stateactionpairset$ exists, there is a state $\state \in \submdp{\states}{\stateactionpairset}$ that does not reach the set $\states \setminus \submdp{\states}{\stateactionpairset}$ almost surely.
We consider the MDP $\mdp^\stateactionpairset$, a mild extension of $\submdp{\mdp}{\stateactionpairset}$ given by
$\mdp^\stateactionpairset = (\submdp{\states}{\stateactionpairset} \uplus \set{\sinit^\stateactionpairset, \state_{\bot}^\stateactionpairset}, \actions \uplus \set{\actinit, \actbot}, \transitions^\stateactionpairset, \sinit^\stateactionpairset)$, where 
$\transitions^\stateactionpairset$ extends $\submdp{\transitions}{\stateactionpairset}$ such that $\transitions^\stateactionpairset(\state_{\bot}^\stateactionpairset, \actbot, \state_{\bot}^\stateactionpairset) = 1$ and 
$\fa{\state \in \submdp{\states}{\stateactionpairset}}$
\begin{itemize}
	\item 
	$ \transitions^\stateactionpairset(\sinit^\stateactionpairset, \actinit, \state) = 1 / |\submdp{\states}{\stateactionpairset}|$,\quad
	$\transitions^\stateactionpairset(\state, \actbot, \state_\bot^\stateactionpairset) = 1 $, and 
	\item $\fa{\action \in \set{\hat{\action} \in \act{\state} \mid \tuple{\state,\hat{\action}} \notin \stateactionpairset}} \transitions^\stateactionpairset(\state,\action, \state_{\bot}^\stateactionpairset)  = 1$.
\end{itemize}
Lines~\ref{lp:ec:flow:inout} and~\ref{lp:ec:flow:end} reflect the equation system from~\Cref{lem:flow} for MDP $\mdp^\stateactionpairset$ and $\stateszero = \set{\state_{\bot}}$.
Additionally, Lines~\ref{lp:ec:flow:start} and~\ref{lp:ec:flow:outvar} exclude negative solutions and assert $\lpsolof{\lpecflowvarsaj} = 0$ if $\lpsolof{\lpactionvarsa} = 0$ and $\lpsolof{\lpecflowvar{\state}{\actbot}{\objj}} = 0$ if $\lpsolof{\lpecvarsj} = 0$ for any solution $\lpsol$.
Hence, for states $\state \in \submdp{\states}{\stateactionpairset}$ where $\lpsolof{\lpecvarsj} = 0$, the strategy $\sched$ encoded by the variables $\lpactionvarsa$ coincides with the strategy considered in \Cref{lem:flow}.
Assume that solution $\lpsol$ yields a BSCC within the states of $\stateactionpairset$ in $\mc{\mdp}{\sched}$ and therefore also a BSCC in $\mc{(\mdp^\stateactionpairset)}{\sched}$.
Since $\state_{\bot}^\stateactionpairset$ has to be reached almost surely in $\mdp^\stateactionpairset$ (cf. \Cref{lem:flow}), the BSCC has to contain at least one state $s$ with $\lpsolof{\lpecvarsj} = 1$.

In summary, Lines~\ref{lp:ec:flow:start} to~\ref{lp:ec:flow:end} imply that every BSCC $\stateactionpairset' \subseteq \stateactionpairset$ of $\mc{\mdp}{\sched}$ contains at least one state $s$ with $\lpsolof{\lpecvarsj} = 1$. 
Then, with Lines~\ref{lp:ec:graph:start} to~\ref{lp:ec:graph:end} we get that $\lpsolof{\lpecvar{\state'}{\objindex}} = 1$ has to hold for \emph{all} $\state' \in \submdp{\states}{\stateactionpairset'}$.
In $\mc{\mdp}{\sched}$, the set $\stateszeroj \cup \set{\state \mid \lpsolof{\lpecvar{\state}{\objindex}} = 1}$ is therefore reached almost surely and all the states in this set get assigned value 0. In this case, the solution of the equation system from \Cref{lem:mceqsys} becomes unique again.

\begin{theorem}\label{thm:multichainfinite}
	For finite rewards, the constraints in~\Cref{lp:unichainfinite,lp:endcomponents} are feasible iff
	$\point \in \psach{\mdp}{\multiobjquery}$.
\end{theorem}

\begin{figure}[t]
	\centering
	\begin{subfigure}[b]{0.37\textwidth}
		\centering
		\begin{tikzpicture}[xscale=0.8,->,>=stealth',shorten >=1pt,auto,node
		distance=2.5cm,bend angle=45, scale=0.6, font=\normalsize]
		\path[use as bounding box] (-3.2,-0.7) rectangle (6,1.1);
		\tikzstyle{p1}=[draw,circle,text centered,minimum size=6mm,text width=3mm]
		\tikzstyle{p2}=[draw,rectangle,text centered,minimum size=6mm,text width=3mm]
		\tikzstyle{act}=[fill,circle,inner sep=1pt,minimum size=1.5pt, node distance=1cm]    \tikzstyle{empty}=[text centered, text width=15mm]
		\node[p1] (1) at (0,0) {$\state_0$};
		\node[p1,accepting] (2) at (4,0) {$\state_1$};
		\node[act] (1a) at (2,0) {};
		\node[act] (1b) at ($(1) - (1.75, 0)$) {};
		\node[act] (2a) at ($(2.east) + (1, 0)$) {};
		\node[empty] at ($(1a)+(0,0.4)$) {\scriptsize $\beta : 1$};
		\node[empty] at ($(1b)-(0.8,0)$) {\scriptsize $\alpha : 0$};
		\node[empty] at (0,-1.5) {};
		\node[]  (init) at (-1.2,1.2) {};
		\path[-latex']
		(init) edge (1)
		(1a) edge node[below,xshift=0mm]{\scriptsize $ $} (2);
		\draw [->] (1) -- (1a);
		\draw [->] (1a) -- (2);
		\draw [->] (1) to [out=145,in=70,looseness=1] (1b);
		\draw [->] (1b) to [out=-70,in=-145, looseness=1] (1);
		\draw [->] (2) to [out=35,in=110, looseness=1] (2a);
		\draw [->] (2a) to [out=-110,in=-35,looseness=1] (2);
		\end{tikzpicture}
		\caption{Multichain MDP}
		\label{fig:multichain}
	\end{subfigure}
	\hfill
	\begin{subfigure}[b]{0.62\textwidth}
		\centering
		\begin{tikzpicture}[xscale=0.8,->,>=stealth',shorten >=1pt,auto,node
		distance=2.5cm,bend angle=45, scale=0.6, font=\normalsize]
		\tikzstyle{p1}=[draw,circle,text centered,minimum size=6mm,text width=3mm]
		\tikzstyle{p2}=[draw,rectangle,text centered,minimum size=6mm,text width=3mm]
		\tikzstyle{act}=[fill,circle,inner sep=1pt,minimum size=1.5pt, node distance=1cm]    \tikzstyle{empty}=[text centered, text width=15mm]
		\path[use as bounding box] (-1.3,-0.7) rectangle (13.8,1.1);
		\node[p1] (1) at (0,0) {$\state_1$};
		\node[p1] (2) at (4,0) {$\state_2$};
		\node[p1, accepting] (3) at (8,0) {$\state_3$};
		\node[p2] (4) at (12,0) {$\state_4$};
		\node[act] (1a) at (2,0) {};
		\node[act] (2a) at (6,0.5) {};
		\node[act] (2b) at (6,-0.5) {};
		\node[act] (3a) at (9.5,0) {};
		\node[act] (4a) at ($(4.east) + (1,0)$) {};
		\node[empty] at ($(2a)+(0,0.4)$) {\scriptsize $\alpha$};
		\node[empty] at ($(2b)-(0,0.4)$) {\scriptsize $\beta$};
		\node[]  (init) at (-1.2,1.2) {};
		\path[-latex']
		(init) edge (1)
		(1) edge (1a)
		(1a) edge node[above,xshift=0mm]{\scriptsize $0.5$} (2);
		(2) edge (2b)
		(2a) edge (3)
		(2b) edge (4)
		\draw [->] (1a) to[out=90,in=45] node[above,xshift=0mm]{\scriptsize $0.5$} (1);
		\draw [->] (2) -- (2a);
		\draw [->] (2) -- (2b);
		\draw [->] (2a) -- (3);
		\draw [->] (2b) to [bend right=15] (4);
		\draw [->] (3) to [out=35,in=110, looseness=1] (3a);
		\draw [->] (3a) to [out=-110,in=-35,looseness=1] (3);
		\draw [->] (4) to [out=35,in=110, looseness=1] (4a);
		\draw [->] (4a) to [out=-110,in=-35,looseness=1] (4);
		\end{tikzpicture}
		\caption{MDP with nontrivial memory requirements}
		\label{fig:memory}
	\end{subfigure}
	\caption{MDPs referenced in \Cref{ex:multichain,ex:nontrivial-memory}.}
\end{figure}

\subsection{Extension to Infinite Rewards}\label{sec:achievability:infinite}
Our approach can be modified to allow $\psma$ instances where infinite expected reward can be collected, \ie where \Cref{as:rewardfiniteness} does not hold.
Infinite reward can be collected if we cycle through an EC of $\mdp$ that contains a transition with positive reward.
Such instances are of practical interest as this often corresponds to strategies that do not accomplish a certain goal (\eg a robot that stands still and therefore requires infinite time to finish its task).

We sketch the necessary modifications. More details are in \reftr{B}.
Let $\statesinf$ be the set of states where every pure strategy induces infinite reward for at least one minimizing objective.
	To ensure that the \milp instance has a (real-valued) solution, we consider the sub-MDP of $\mdp$ obtained by removing $\statesinf$.
	
If infinite reward can be collected in an EC, it should not be considered in~\Cref{lp:endcomponents}. We therefore let $\stateactionpairset$ range over maximal ECs that only consist of (a) states in $\statesmaybej$ and (b) transitions with reward 0.
	
	The upper bounds $\uppervalueboundsj$ for the maximal expected rewards at each state can not be set to $\infty$.
	However, for the encoding it suffices to compute values that are sufficiently large.
	However, we remark that in practice our approach from \reftr{B} can lead to very large values, yielding numerical instabilities.
	
	For maximizing objectives, we introduce one additional objective which, in a nutshell, checks that the probability to reach a 0-reward BSCC is below 1. If this is the case, there is a positive probability to reach a BSCC in which infinitely many reward can be collected.
	
%

\section{Computing the Pareto Front}
\label{sec:pareto}

Our next goal is to compute the \emph{pure stationary Pareto front} $\pspareto{\mdp}{\multiobjquery}$ for MDP $\mdp$ and multi-objective query $\multiobjquery$.
This set can be very large, in particular if the objectives are strongly conflicting with many different tradeoffs.
In the worst case, every pure stationary strategy induces a point $\point \in \pspareto{\mdp}{\multiobjquery}$ (\eg for $\multiobjquery = \tuple{\rewobjop_\le(\rewstruct \eventually \goalstates), \rewobjop_\ge(\rewstruct \eventually \goalstates)}$).
We try to find an approximation of $\pspareto{\mdp}{\multiobjquery}$.

\begin{definition}
Let $\precision \in \precisiondomain$. An \emph{$\precision$-approximation} of $\points \subseteq \pointdomain$ is a pair $\approxtuple$ with $\underapprox \subseteq \points \subseteq \overapprox$ and $\fa{\point \in \points} \ex{\point' \in \underapprox \cup (\pointdomain \setminus \overapprox)} |\point-\point'| \le \precision$.
\end{definition}
\begin{nproblem}[framed]{Pure Stationary Pareto Approximation Problem (\pspapprox)}
	Input: & MDP $\mdp$, $\numobj$-dimensional multi-objective query $\multiobjquery$, precision 
	 $\precision \in \precisiondomain$
	 such that $\pspareto{\mdp}{\multiobjquery} \subseteq \rr^\numobj$ \\
	Output: & An $\precision$-approximation of $\closure{\multiobjquery}{\pspareto{\mdp}{\multiobjquery}}$
\end{nproblem}
For simplicity, we only consider inputs that satisfy restriction \Cref{as:rewardfiniteness}, \ie for $\objj = \rewobjop_{\objrelj}(\rewstructj \eventually \goalstatesj)$ there is $\uppervaluebound{}{\objindex} \neq \infty$ such that $\fa{\sched \in \psscheds{\mdp}} \uppervaluebound{}{\objindex} \ge \expval{\mdp}{\sched}(\rewstructj \eventually \goalstatesj)$.
Ideas of \Cref{sec:achievability:infinite} can be used for some other inputs. An all-embracing treatment of infinite rewards, in particular for maximizing $\objj$, is subject to future work.

Our approach for \pspapprox successively divides the solution space into candidate regions.
For each region $\region$ (initially, let $\region = [0,\uppervaluebound{}{1}] \times \dots \times [0, \uppervaluebound{}{\numobj}]$), we use the \milp encoding from \Cref{sec:achievability} with an optimization function to find a point $\point \in \region \cap \pspareto{\mdp}{\multiobjquery}$ (or find out that no such point exists).
The region $\region$ is divided into (i) an achievable region $\region_\mathrm{A} \subseteq \psach{\mdp}{\multiobjquery}$, (ii) an unachievable region $\region_\mathrm{U} \subseteq \rr^\numobj \setminus \psach{\mdp}{\multiobjquery}$, (iii) further candidate regions $\region_1, \dots, \region_n$ that are analyzed subsequently, and (iv) the remaining area $\region \setminus (\region_\mathrm{A} \cup \region_\mathrm{U} \cup \region_1 \cup \dots \cup \region_n)$ which does not require further analysis as we are only interested in an $\precision$-approximation.
The procedure stops as soon as no more candidate regions are found.
%
\begin{example}
	\label{ex:pareto}
	\Cref{fig:pareto} sketches the approach for an MDP $\mdp$ and a query $\multiobjquery$ with two maximizing objectives.
	We maintain a set of achievable points (light green) and a set of unachievable points (red).
	Initially, our candidate region corresponds to $\region_1 = [0,\uppervaluebound{}{1}] \times [0, \uppervaluebound{}{2}]$ given by the white area in~\Cref{fig:pareto:init}.
	We consider the direction vector $\dirvec[1]$  which is orthogonal to the line connecting $\tuple{\uppervaluebound{}{1},0}$ and $\tuple{0,\uppervaluebound{}{2}}$.
	To find some point $\point \in \pspareto{\mdp}{\multiobjquery} \cap \region_1$, we solve the \milp resulting from the constraints as in~\Cref{sec:achievability}, the constraint $\tuple{\lpvalvar{\sinit}{1}, \lpvalvar{\sinit}{2}} \in \region_1$, and the optimization function $\dirvec[1] \cdot \tuple{\lpvalvar{\sinit}{1}, \lpvalvar{\sinit}{2}}$.
	\Cref{fig:pareto:vertex} shows the obtained point $\point_1 \in \region_1$.
	Since $\point_1$ is achievable, we know that any point in $\closure{\multiobjquery}{\set{\point_1}}$ has to be achievable as well.
	Moreover, the set $\set{ \point \in \region_1 \mid \dirvec[1] \cdot \point > \dirvec[1] \cdot \point_1}$ indicated by the area above the diagonal line in~\Cref{fig:pareto:vertex} can not contain an achievable point.
	The gray areas do not have to be checked in order to obtain an $\precision$-approximation.
	We continue with $\region_2$ indicated by the white area and the direction vector $\dirvec[2]$, orthogonal to the line connecting $\tuple{0,\uppervaluebound{}{2}}$ and $\point_1$.
	As before, we solve an \milp now yielding the point $\point_2$ in \Cref{fig:pareto:inner}.
	We find achievable points $\closure{\multiobjquery}{\set{\point_2}}$ but no further unachievable points.
	The next iteration considers candidate region $\region_3$ and direction vector $\dirvec[1]$, yielding point $\point_3$ shown in \Cref{fig:pareto:final}. 
	The trapezoidal area is added to the unachievable points whereas $\closure{\multiobjquery}{\set{\point_3}}$ is achievable.
	Finally, we check $\region_4$ for which the corresponding $\milp$ instance is infeasible, \ie $\region_4$ is unachievable.
\end{example}
\begin{figure}[t]
	\newcommand{\paretoplotscale}{0.41}
	\newcommand{\paretoplotwidth}{0.243\textwidth}
	\newcommand{\paretoplotxmax}{7}
	\newcommand{\paretoplotymax}{7}
	\centering
	\begin{subfigure}[b]{\paretoplotwidth}
		\centering
		\begin{tikzpicture}[scale=\paretoplotscale]
		\path[use as bounding box] (0,0) rectangle (\paretoplotxmax,\paretoplotymax);
		\path[nonach] (0,6.5) rectangle (\paretoplotxmax,\paretoplotymax);
		\path[nonach] (6.5,0) rectangle (\paretoplotxmax,\paretoplotymax);
		\path[->]
		(0,0) edge (\paretoplotxmax,0)
		(0,0) edge (0,\paretoplotymax)
		;
		\node[point] at (6.5,0) (u1) {};
		\node[] at (5,0.3) (lu1) {\scriptsize$\tuple{\uppervaluebound{}{1},0}$};
		\node[] at (1.5,6) (lu2) {\scriptsize$\tuple{0,\uppervaluebound{}{2}}$};
		\node[point] at (0,6.5) (u2) {};
		\node[] at (2,2) (){\scriptsize$\region_1$};
		\path[dirvec] (3.25,3.25) edge node[below right] {$\dirvec[1]$} +(2,2);
		\path[auxline] (6.5,0) edge (0,6.5);
		\end{tikzpicture}
		\caption{}
		\label{fig:pareto:init}
	\end{subfigure}
	\begin{subfigure}[b]{\paretoplotwidth}
		\centering
		\begin{tikzpicture}[scale=\paretoplotscale]
		\path[use as bounding box] (0,0) rectangle (\paretoplotxmax,\paretoplotymax);
		\path[nonach] (0,6.5) rectangle (\paretoplotxmax,\paretoplotymax);
		\path[nonach] (6.5,0) rectangle (\paretoplotxmax,\paretoplotymax);
		\path[nonach] (6.5,1.5)--(1.5,6.5)--(6.5,6.5)--cycle;
		\path[ach] (0,0) rectangle (6,2);
		\path[dk] (6,2)--(6,0)--(6.5,0)--(6.5,1.5)--cycle ;
		\path[dk] (6,2)--(5.5,2.5)--(0,2.5)--(0,2)--cycle ;
		\path[->]
		(0,0) edge (\paretoplotxmax,0)
		(0,0) edge (0,\paretoplotymax)
		;
		\node[point,label={[xshift=-3pt]30:{$\point_1$}}] at (6,2) (p1) {};
		\node[] at (2,3.5) (){\scriptsize$\region_2$};
		\path[dirvec] (4,4) edge node[below right] {$\dirvec[1]$} +(2,2);
		\path[dirvec] (3,4.25) edge node[above left] {$\dirvec[2]$} +(2,2.66666);
		\path[auxline] (p1) edge (0,6.5);
		\end{tikzpicture}
		\caption{}
		\label{fig:pareto:vertex}
	\end{subfigure}
	\begin{subfigure}[b]{\paretoplotwidth}
		\centering
		\begin{tikzpicture}[scale=\paretoplotscale]
		\path[use as bounding box] (0,0) rectangle (\paretoplotxmax,\paretoplotymax);
		\path[nonach] (0,6.5) rectangle (\paretoplotxmax,\paretoplotymax);
		\path[nonach] (6.5,0) rectangle (\paretoplotxmax,\paretoplotymax);
		\path[nonach] (6.5,1.5)--(1.5,6.5)--(6.5,6.5)--cycle;
		\path[ach] (0,0) rectangle (6,2);
		\path[ach] (0,0) rectangle (1.5,6.5);
		\path[dk] (6,2)--(6,0)--(6.5,0)--(6.5,1.5)--cycle ;
		\path[dk] (6,2)--(5.5,2.5)--(1.5,2.5)--(1.5,2)--cycle ;
		\path[dk] (1.5,6.5)--(1.5,2)--(2,2)--(2,6)--cycle ;
		\path[->]
		(0,0) edge (\paretoplotxmax,0)
		(0,0) edge (0,\paretoplotymax)
		;
		\node[point,label={[xshift=-3pt]30:{$\point_1$}}] at (6,2) (p1) {};
		\node[point,label=0:{$\point_2$}] at (1.5,6.5) (p2) {};
		\node[] at (3.5,3.5) (){\scriptsize$\region_3$};
		\path[dirvec] (4,4) edge node[below right] {$\dirvec[1]$} +(2,2);
		\end{tikzpicture}
		\caption{}
		\label{fig:pareto:inner}
	\end{subfigure}
	\begin{subfigure}[b]{\paretoplotwidth}
		\centering
		\begin{tikzpicture}[scale=\paretoplotscale]
		\path[use as bounding box] (0,0) rectangle (\paretoplotxmax,\paretoplotymax);
		\path[nonach] (0,6.5) rectangle (\paretoplotxmax,\paretoplotymax);
		\path[nonach] (6.5,0) rectangle (\paretoplotxmax,\paretoplotymax);
		\path[nonach] (6.5,1.5)--(1.5,6.5)--(6.5,6.5)--cycle;
		\path[nonach,border] (5.5,2.5)--(4.5,2.5)--(2,5)--(2,6)--cycle;
		\path[nonach,border] (2.5,2.5)--(4.5,2.5)--(2.5,4.5)--cycle;
		\path[ach] (0,0) rectangle (6,2);
		\path[ach] (0,0) rectangle (1.5,6.5);
		\path[ach] (0,0) rectangle (2.5,4.5);
		\path[dk] (6,2)--(6,0)--(6.5,0)--(6.5,1.5)--cycle ;
		\path[dk] (6,2)--(5.5,2.5)--(2.5,2.5)--(2.5,2)--cycle ;
		\path[dk] (3,2)--(3,4)--(2.5,4.5)--(2.5,2)--cycle ;
		\path[dk] (1.5,6.5)--(1.5,4.5)--(2,4.5)--(2,6)--cycle ;
		\path[dk] (1.5,4.5)--(2.5,4.5)--(2,5)--(1.5,5)--cycle ;
		\path[->]
		(0,0) edge (\paretoplotxmax,0)
		(0,0) edge (0,\paretoplotymax)
		;
		\node[point,label={[xshift=-3pt]30:{$\point_1$}}] at (6,2) (p1) {};
		\node[point,label=0:{$\point_2$}] at (1.5,6.5) (p2) {};
		\node[point,label=-150:{$\point_3$}] at (2.5,4.5) (p3) {};
		\node[] at (3.5,3) (){\scriptsize$\region_4$};
		\path[dirvec] (3,4) edge node[below right] {$\dirvec[1]$} +(2,2);
		
		\end{tikzpicture}
		\caption{}
		\label{fig:pareto:final}
	\end{subfigure}
\vspace{-18pt}
	\caption{Example exploration of achievable points.}
	\label{fig:pareto}
\end{figure}
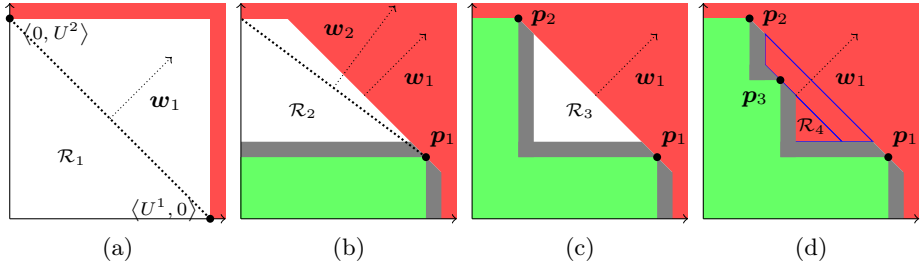
The ideas sketched above  can be lifted to $\numobj > 2$ objectives.
Inspired by~\cite[Alg. 4]{FKP12}, we choose direction vectors that are orthogonal to the convex hull of the achievable points found so far.
In fact, for total reward objectives we can apply the approach of~\cite{FKP12} to compute the points in $\pspareto{\mdp}{\multiobjquery} \cap \pareto{\mdp}{\multiobjquery}$ first and only perform \milp-solving for the remaining regions.
As the distance between two found points $\point,\point'$ is at least $|\point-\point'|\ge \precision$, we can show that our approach terminates after finding at most $\prod_{\objindex} \uppervaluebound{}{\objindex} / \tupleaccess{\precision}{\objindex}$ points.
Other strategies for choosing direction vectors are possible and can strongly impact performance.
\section{Bounded Memory}
\label{sec:boundedmem}
For \gma, it is necessary and sufficient to consider strategies that require memory exponential in the number of objectives \cite{EKVY07,FKP12,RRS17} by storing which goal state set has been reached already. 
In contrast, restricting to pure (but not necessarily stationary) strategies imposes nontrivial memory requirements that do not only depend on the number of objectives, but also on the point that is to be achieved.
\begin{example}\label{ex:nontrivial-memory}
	Let $\mdp$ be the MDP in \Cref{fig:memory} and $\multiobjquery = \tuple{\reachobjop_\geq\left(\eventually\goalstates_\varocircle\right), \reachobjop_\geq\left(\eventually\goalstates_\square\right)}$.
	The point $\point_k = \tuple{0.5^k, 1 {-} 0.5^k}$ for $k \in \nn$ is achievable by taking $\alpha$ with probability $0.5^k$.
	$\point_k$ is also achievable with the \emph{pure} strategy $\sched_k$ where  $\sched_k(\finpath) = \alpha$ iff $\lengthofpath{\finpath} \ge k$. $\sched_k$ uses $k$ memory states.
	Pure strategies with fewer memory states do not suffice.
\end{example}
We search for pure strategies with \emph{bounded memory}.
For an MDP $\mdp$ and $K > 0$, let $\pschedsK{\mdp}{K}$ denote the set of pure $K$-memory strategies, \ie any $\sched \in \pschedsK{\mdp}{K}$ can be represented by a \emph{Mealy machine} using up to $K$ states (c.f.~\reftr{D}).
For a query $\multiobjquery$, let $\pKach{\mdp}{\multiobjquery}{K}$ be the set of points achievable by some $\sched \in \pschedsK{\mdp}{K}$\iftoggle{TR}{, and let $\pKpareto{\mdp}{\multiobjquery}{K}$ be the \emph{pure $K$-memory} Pareto front.}{.}
\begin{nproblem}[framed]{Pure Bounded Multi-objective Achievability Problem (\pbma)}
	Input: & MDP $\mdp$, multi-objective query $\multiobjquery$, memory bound $K\!$, point $\point \in \pointdomain$\\
	Output: & Yes iff $\point \in \pKach{\mdp}{\multiobjquery}{K}$
\end{nproblem}
\noindent The pure bounded Pareto approximation problem is defined similarly.
We reduce a \pbma instance to an instance for \psma.
The idea is to incorporate a memory structure of size $K$ into $\mdp$ and then construct a pure stationary strategy in this product MDP (see, \eg \cite{JJWQWKB18} for a similar construction).%
\iftoggle{TR}{%
\begin{definition}[Memory structure]
	A \emph{memory structure} of size $K > 0$ is a tuple $\ndmemory{K} = \ndmemorytuple$ with $|M| = K$, initial memory state $\minit \in \memorystates$, and nondeterministic memory update function $\ndmemoryupdate \colon \memorystates \to 2^\memorystates \setminus \emptyset$.%
\end{definition}
\begin{definition}[Memory product]\label{def:memprod}
The product of MDP $\mdp = \mdptuple$ and memory structure $\ndmemory{K} = \ndmemorytuple$ is given by the MDP $\mdpndmemory{K} = \tuple{\states {\times} \memorystates, \actions {\times} M, \mdpndmemorytransitions, \tuple{ \sinit, \minit}}$, where  
    for $\state, \state' \in \states$, $m, m' \in \memorystates$, and $\action \in \actions$:
   $\mdpndmemorytransitions( \tuple{\state, \memorystate}, \tuple{\action, \memorystate'}, \tuple{ \state', \memorystate'}) = \transitions\left( \state, \action, \state' \right) \cdot \iverson{\memorystate' \in \ndmemoryupdate\left( m \right)}$.
\end{definition}
Intuitively, $\ndmemoryupdate(m)$ gives the possible successors of memory state $m\in M$.
The memory product enriches $\mdp$ with a new level of nondeterminism corresponding to the choice of the next memory state in $\ndmemory{K}$.
We set up an equivalence between $K$-memory strategies for $\mdp$ and stationary strategies for $\mdpndmemory{K}$ by considering \emph{complete} memory structures, \ie memory structures with $\fa{\memorystate \in \memorystates} \ndmemoryupdate(\memorystate) = \memorystates$.
%
\begin{restatable}{lemma}{schedEquivLemma}\label{lem:bounded-mem-equiv}
	\iftoggle{TR}{Let $\mdp = \mdptuple$ be an MDP, $\ndmemory{K} = \ndmemorytuple$ be a complete nondeterministic memory structure, and $\mdpndmemory{K} = \mdpndmemorytuple$ be their product.}{}
	There are equivalence relations between (i) strategies $\sched \in \pschedsK{\mdp}{K}$ and $\sched' \in \psscheds{\mdpndmemory{K}}$, and (ii) paths compliant with $\sched$ and $\sigma'$ such that equivalence (i) preserves the probability of paths in equivalence relation (ii).%
\end{restatable}
\begin{restatable}{corollary}{corBoundedMultiQuery}\label{cor:bounded-to-stationnary}
	\iftoggle{TR}{%
	Let $\multiobjquery = \tuple{\obj_1, \dots, \obj_\numobj}$  and $\multiobjquery' = \tuple{\obj'_1, \dots, \obj'_\numobj}$ be a multi-objective queries for $\mdp$ and $\mdpndmemory{K}$ respectively such that $\fa{\objindex}
	\obj'_\objindex = \rewobjop_\objrel\left(\rewstruct' \eventually \goalstates' \right)$ iff $\obj_\objindex = \rewobjop_\objrel(\rewstruct \eventually \goalstates)$,
	with $\goalstates' = \goalstates \times \memorystates$ if $\goalstates \neq \emptyset$ and  $\goalstates' = \emptyset$ otherwise, and $\rewstruct'((s, m), (a, m'), (s', m')) = \rewstruct(s, a, s')$ for $\state, \state' \in \states$, $\action \in \actions$, and $\memorystate, \memorystate' \in \memorystates$.
	Then, $\pKach{\mdp}{\multiobjquery}{K} = \psach{\mdpndmemory{K}}{\multiobjquery'}$.}{%
	$\pKach{\mdp}{\multiobjquery}{K} = \psach{\mdpndmemory{K}}{\multiobjquery'}$ for suitable extension $\multiobjquery'$ of $\multiobjquery$.}
\end{restatable}
\noindent \pbma can thus be decided by solving \psma for $\mdpndmemory{K}$.
Proofs and further details are given in \reftr{D}.
The strategies can be further simplified by considering non-full memory structures, \eg a memory structure that only allows counting. 

\begin{remark}[Memory patterns]
	Complete memory structures $\ndmemory{K}= \ndmemorytuple$ are in general necessary for \pbma
	However, according to the instance, it could be sufficient to consider other nondeterministic memory update functions $\ndmemoryupdate$, implying less transitions in $\mdpndmemory{K}$ and considerably reducing the number of constraints of the \milp. For example, for the MDP of \Cref{fig:memory}, it is sufficient to consider a nondeterministic memory structure such that $\delta(m_i) = \{m_i, m_{i+1}\}$ if $i < K$ and $\delta(m_i) = \{m_{i}\}$ otherwise, to compute a pure $K$-memory strategy that achieves 
	$\point = \left(\left( \frac{1}{2} \right)^{K - 1}, 1 - \left( \frac{1}{2} \right)^{K - 1} \right)$ to satisfy $\multiobjquery = \tuple{\reachobjop_\geq\left(\goalstates_\varocircle\right), \reachobjop_\geq\left(\goalstates_\square\right)}$ (cf. Example~\ref{ex:nontrivial-memory}). Note however that in general, with non-complete memory structure $\ndmemory{K}$,
	we only have $\psach{\mdp}{\multiobjquery} \subseteq \psach{\mdpndmemory{K}}{\multiobjquery'} \subseteq \pKach{\mdp}{\multiobjquery}{K}$.
\end{remark}
}{
The set of strategies can be further refined by considering  \eg a memory structure that only allows counting or that only remembers visits of goal states. 
See \reftr{D} for details.}

\begin{table}[t]
	\caption{Results for stationary strategies.}
	\label{tab:res}
	\adjustbox{max width=\textwidth}{%
		\begin{tabular}{ll?crrr|rr|rr?crrr|rr|rr}
			\multicolumn{2}{l?}{Bench-} & \multicolumn{4}{c|}{Instance 1} & \multicolumn{2}{c|}{$\varepsilon {=} 0.01$} & \multicolumn{2}{c?}{$\varepsilon {=} 0.001$} & \multicolumn{4}{c|}{Instance 2} & \multicolumn{2}{c|}{$\varepsilon {=} 0.01$} & \multicolumn{2}{c}{$\varepsilon {=} 0.001$}\\
			\multicolumn{1}{c}{mark} & \multicolumn{1}{c?}{$\numobj$} 
			&Par. & \multicolumn{1}{c}{$|\states|$} & \multicolumn{1}{c}{$\%\stateactionpairset$} & \multicolumn{1}{c|}{$\overline{\act{}}$} & \multicolumn{1}{c}{Time} & \multicolumn{1}{c|}{$|\points|$} &   \multicolumn{1}{c}{Time} & \multicolumn{1}{c?}{$|\points|$}  
			&Par. & \multicolumn{1}{c}{$|\states|$} & \multicolumn{1}{c}{$\%\stateactionpairset$} & \multicolumn{1}{c|}{$\overline{\act{}}$} & \multicolumn{1}{c}{Time} & \multicolumn{1}{c|}{$|\points|$} &   \multicolumn{1}{c}{Time} & \multicolumn{1}{c}{$|\points|$}  
			\\\thickhline
			\benchmark{dpm} & 2$^*$ & 2 & 1272 & 32 & 3.2 & 17 & 37 & 315 & 377 & 3 & 1696 & 30 & 3.2 & 82 & 30 & \multicolumn{2}{c}{\timeout}\\
\benchmark{eajs} & 2$^*$ & 2-3 & 689 & 0 & 1.2 & 5 & 23 & 45 & 202 & 3-6 & $2{\cdot}10^4$ & 0 & 1.2 & 201 & 52 & 3787 & 375\\
\benchmark{jobs} & 3$^*$ & 3-2 & 17 & 0 & 1.1 & 3 & 3 & 2 & 3 & 5-2 & 117 & 0 & 1.5 & 2042 & 76 & \multicolumn{2}{c}{\timeout}\\
\benchmark{mutex} & 3$^*$ & 1 & 1795 & 36 & 2.2 & \multicolumn{2}{c|}{\timeout} & \multicolumn{2}{c?}{\timeout} & 2 & $1{\cdot}10^4$ & 33 & 2.3 & \multicolumn{2}{c|}{\timeout} & \multicolumn{2}{c}{\timeout}\\
\benchmark{polling} & 2 & 2-2 & 233 & 86 & 1.5 & 6 & 5 & 23 & 6 & 3-2 & 990 & 84 & 1.8 & 299 & 5 & \multicolumn{2}{c}{\timeout}\\
\benchmark{rg} & 2$^*$ & 2-1-20 & 2173 & 14 & 2.9 & 5 & 5 & 12 & 5 & 5-2-50 & $3{\cdot}10^4$ & 5 & 3.1 & 496 & 27 & \multicolumn{2}{c}{\timeout}\\
\benchmark{rover} & 2$^*$ & 2500 & $2{\cdot}10^4$ & 0 & 1.2 & 110 & 47 & 417 & 251 & 5000 & $4{\cdot}10^4$ & 0 & 1.2 & 258 & 47 & 3105 & 472\\
\benchmark{serv} & 2$^*$ &  & $5{\cdot}10^4$ & 93 & 1.9 & 1828 & 38 & \multicolumn{2}{c?}{\timeout} &   &   &   &   &   &   &   &  \\
\benchmark{str} & 2$^*$ & 30 & 1426 & 0 & 1.3 & 11 & 21 & 822 & 218 & 500 & $4{\cdot}10^5$ & 0 & 1.3 & 2428 & 17 & \multicolumn{2}{c}{\timeout}\\
\benchmark{team2} & 2$^*$ & 2 & 1847 & 24 & 1.2 & 2 & 5 & 2 & 5 & 3 & $1{\cdot}10^4$ & 21 & 1.2 & 18 & 43 & \multicolumn{2}{c}{\memout}\\
\benchmark{team3} & 3$^*$ & 2 & 1847 & 24 & 1.2 & 165 & 15 & 166 & 15 & 3 & $1{\cdot}10^4$ & 21 & 1.2 & \multicolumn{2}{c|}{\timeout} & \multicolumn{2}{c}{\timeout}\\
\benchmark{uav} & 2$^*$ & 750 & $2{\cdot}10^5$ & 29 & 1.6 & 400 & 39 & 5799 & 332 & 1000 & $4{\cdot}10^5$ & 31 & 1.8 & 3546 & 36 & \multicolumn{2}{c}{\timeout}\\
\benchmark{wlan} & 2$^*$ & 0 & 2954 & 0 & 1.3 & 160 & 16 & \multicolumn{2}{c?}{\timeout} & 2 & $3{\cdot}10^4$ & 0 & 1.3 & 6728 & 23 & \multicolumn{2}{c}{\timeout}
		\end{tabular}%
	}
\end{table}
\section{Evaluation}
\label{sec:evaluation}
We implemented our approach for \pspapprox in the model checker \tool{Storm}~\cite{DJKV17} using \tool{Gurobi}~\cite{gurobi} as back end for \milp-solving.
The implementation takes an MDP (\eg in \tool{Prism} syntax), a multi-objective query, and a precision $\varepsilon > 0$ as input and computes an $\precision$-approximation of the Pareto front.
Here, we set $\tupleaccess{\precision}{\objindex} = \varepsilon \cdot \delta_\objindex$, where $\delta_\objindex$ is the difference between the maximal and minimal achievable value for objective $\objj$. 
We also support reward objectives for Markov automata via~\cite{QJK17}.
The computations within \tool{Gurobi} might suffer from numerical instabilities.
To diminish their impact, we use the \emph{exact} engine of \tool{Storm} to confirm for each $\milp$ solution that the encoded strategy achieves the encoded point.
However, sub-optimal solutions returned by \tool{Gurobi} may still yield inaccurate results.

We evaluate our approach on 13 multi-objective benchmarks from~\cite{FKP12,HJKQ18,QJK17}, each considering one or two  parameter instantiations.
Application areas range over scheduling (%
\benchmark{dpm}~\cite{Qiu99},
\benchmark{eajs}~\cite{BDDKK14},
\benchmark{jobs}~\cite{BDF81},
\benchmark{polling}~\cite{Sri91}),
planning (%
\benchmark{rg}~\cite{BN08},
\benchmark{rover}~\cite{HJKQ18},
\benchmark{serv}~\cite{LPH17},
\benchmark{uav}~\cite{FWHT15}),
and protocols (%
\benchmark{mutex}~\cite{QJK17},
\benchmark{str}~\cite{QJK17},
\benchmark{team}~\cite{CKPS11},
\benchmark{wlan}~\cite{KNP12}).

The results for pure stationary strategies are summarized in \Cref{tab:res}.
For each benchmark we denote the number of objectives $\numobj$ and whether the alternative encoding from~\Cref{sec:achievability:altencoding} has been applied ($^*$).
For each parameter instantiation (Par.), the number of states ($|S|$), the percentage of the states that are contained in an end component ($\%\stateactionpairset$), and the average number of available actions at each state ($\overline{\act{}}$) are given.
For each precision $\varepsilon \in \{0.01,0.001\}$, we then depict the runtime of \tool{Storm} and the number of points on the computed approximation of the Pareto front.
\timeout{} denotes that the approach did not terminate within 2 hours, \memout{} denotes insufficient memory (16\,GB).
All experiments used 8 cores of an Intel\textsuperscript{\textregistered} Xeon\textsuperscript{\textregistered} Platinum 8160 Processor.

\tool{Storm} is often able to compute pure stationary Pareto fronts, even for models with over 100\,000 states (\eg \benchmark{uav}).
However, the model structure strongly affects the performance. For example, the second instance of \benchmark{jobs} is challenging although it only considers 117 states, a low degree of nondeterminism, and no (non-trivial) end components.
Small increments in the model size can increase runtimes significantly (\eg \benchmark{dpm} or \benchmark{uav}).
If a higher precision is requested, much more points need to be found, which often leads to timeouts.
Similarly, for more than 2 objectives the desired accuracy can often not be achieved within the time limit. The approach can be stopped at any time to report on the current approximation, \eg after 2 hours \tool{Storm} found 65 points for Instance 1 of \benchmark{mutex}.

For almost all benchmarks, the objectives could be transformed to total reward objectives, making the more efficient encoding form~\Cref{sec:achievability:altencoding} applicable.
We plot the runtimes of the two encoding in \Cref{fig:results:scatter}.
The alternative encoding is superior for almost every benchmark.
In fact, the original encoding timed out for many models as indicated at the horizontal line at the top of the figure.

In \Cref{fig:results:plot} we plot the Pareto front for the first \benchmark{polling} instance under 
general strategies (Gen), 
pure 2-memory strategies that can change the memory state exactly once (PM$_2$),
pure strategies that observe which goal state set $\goalstatesj$ has been visited already (PM$_\goalstates$), and
pure stationary strategies (PS).
Adding simple memory structures already leads to noticeable improvements in the quality of strategies.
In particular, PM$_2$ strategies perform quite well, and even outperform PM$_\goalstates$ strategies (which would be optimal if randomization were allowed).

\smallskip\noindent\textit{Data availability.}
The artifact~\cite{artifact} accompanying this paper contains source code, benchmark files, and replication scripts for our experiments.

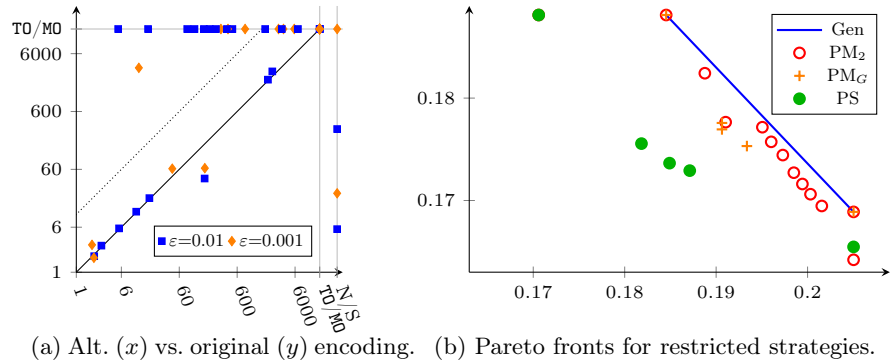
\begin{figure}[t]
\centering
\newlength{\plotsize}
\setlength{\plotsize}{0.42\textwidth}
\begin{subfigure}[b]{\plotsize}
\centering
	\begin{tikzpicture}
	\path[use as bounding box] (-0.6,-0.6) rectangle (4.5,3.5);
\begin{axis}[
width=\plotsize,
height=\plotsize,
axis equal image,
xmin=1,
ymin=1,
ymax=40000,
xmax=40000,
xmode=log,
ymode=log,
axis x line=bottom,
axis y line=left,
xtick={1,6,60,600,6000},
xticklabels={1,6,60,600,6000},
extra x ticks = {16000,32000},
extra x tick labels = {{\timeout/\memout},N/S},
extra x tick style = {grid = major},
ytick={1,6,60,600,6000},
yticklabels={1,6,60,600,6000},
extra y ticks = {16000},
extra y tick labels = {\timeout/\memout},
extra y tick style = {grid = major},
xlabel style={yshift=0.cm,xshift=-0.4cm},
ylabel style={yshift=-0.43cm,xshift=-0.cm},
yticklabel style={font=\scriptsize},
xticklabel style={rotate=290,anchor=west,font=\scriptsize},
legend pos=south east,
legend columns=-1,
legend style={nodes={scale=0.75, transform shape},inner sep=1.5pt,yshift=0.1cm,xshift=-0.3cm},
]
\addplot[
scatter,
only marks,
scatter/classes={
	low={mark=square*,blue,mark size=1.25},
	high={mark=diamond*,orange,mark size=1.75}
},
scatter src=explicit symbolic
]%
table [col sep=semicolon,x=alt,y=classic,meta=type] {scatterdata.csv};
\legend{{$\varepsilon {=} 0.01$},{$\varepsilon {=} 0.001$}};
\addplot[no marks] coordinates {(0.01,0.01) (16000,16000) };
\addplot[no marks, densely dotted] coordinates {(0.01,0.1) (1600,16000)};
\end{axis}
\end{tikzpicture}
\caption{Alt.\ ($x$) vs.\ original ($y$) encoding.}
\label{fig:results:scatter}
\end{subfigure}
\begin{subfigure}[b]{1.2\plotsize}
	\centering
	\begin{tikzpicture}
	\path[use as bounding box] (-0.6,-0.6) rectangle (5.6,3.5);
	\begin{axis}[
width=1.4\plotsize,
height=\plotsize,
xmin=0.163,
ymin=0.163,
ymax=0.189,
xmax=0.209,
axis x line=bottom,
axis y line=left,
yticklabel style={font=\scriptsize},
xticklabel style={font=\scriptsize},
legend pos=north east,
legend columns=1,
legend style={nodes={scale=0.75, transform shape},inner sep=1.5pt,yshift=0cm,xshift=0cm},
]

\addplot[blue,thick] table [col sep=semicolon,x=x,y=y] {polling/gen.csv};
\addplot[only marks, mark=o,red, thick] table [col sep=semicolon,x=x,y=y] {polling/c2.csv};
\addplot[only marks, mark=+,thick,orange] table [col sep=semicolon,x=x,y=y] {polling/goal.csv};
\addplot[only marks, mark=*,green!70!black, thick] table [col sep=semicolon,x=x,y=y] {polling/pos.csv};
\legend{Gen,PM$_2$,PM$_\goalstates$,PS};
\end{axis}
	\end{tikzpicture}
\caption{Pareto fronts for restricted strategies.}
\label{fig:results:plot}
\end{subfigure}
\caption{Comparison of the two encodings (left) and impact of memory (right).}
\end{figure}


\iftoggle{TR}{%
\section{Conclusion}
Finding optimal pure strategies for multi-objective MDPs is NP-hard. 
Yet, such strategies are often desirable, \eg when prescribing medication or designing a product.
We presented an \milp encoding to find optimal pure and stationary strategies on an MDP in which a memory structure can be incorporated.
The encoding is applied to approximate the set of Pareto optimal values.
We successfully compute Pareto fronts for several case studies using our implementation in \tool{Storm}.
Despite the hard nature of the problem, our experiments show feasibility of the approach on practical models with tens of thousands of states.
}{}

\smallskip\noindent\textit{Acknowledgments.}
The authors thank Sebastian Junges for his valuable contributions during early stages of this work.

\bibliographystyle{splncs04}
\bibliography{paper}

\iftoggle{TR}{%
\clearpage
\appendix
\section{Upper Bounds for Expected Number of Visits}
\label{app:A}
\tq{write appendix}
Let $\mdp = \mdptuple$ be an MDP and $\multiobjquery = \tuple{\obj_1, \dots, \obj_\numobj}$ be a multi-objective query over objectives $\objj = \rewobjop_{\objrelj}(\rewstructj \eventually \goalstatesj)$. We use notations as defined in \Cref{sec:achievability}.

For the encodings in \Cref{lp:alternative,lp:endcomponents} we have to compute values $\maxflow{\state}$ for each $\state$ such that the expected number of times a path visits $\state$ from the initial state is at most $\maxflow{\state}$. For this, we only consider pure positional scheduler that reach a given set of sink states $\stateszero$ almost surely.
Formally, we thus require
\[
\fa{\sched \in \psscheds{\mdp}}
\maxflow{\state} \ge 
\iverson{\probmeasure{\mdp}{\sched}(\eventually \stateszero) = 1} \cdot
\expval{\mdp}{\sched}(\rewstruct_{\state,\action} \eventually  \stateszero).
\]
For the case where $\stateszero$ is almost surely reached under any scheduler,~\cite{BKLPW17} provide an efficient, graph based algorithm to compute these values.
Since we have to deal with end components (in particular for the encoding in \Cref{sec:achievability:multichain}, we can not apply the approach of~\cite{BKLPW17} directly.

The idea is to eliminate the end components as in, \eg \cite{HM14}, but in a way that expected visiting times $\expval{\mdp}{\sched}(\rewstruct_{\state,\action} \eventually  \stateszero)$ are over-approximated. For each maximal EC $\stateactionpairset \in \mecs{\submdp{\mdp}{\stateactionpairset_?^\objindex}}$ and for each $\state \in \submdp{\states}{\stateactionpairset}$ within this EC we perform the following steps:
\begin{enumerate}
	\item Compute a lower bound $p>0$ for the probability that starting in $\state$, we leave the EC without visiting $\state$ again.
	For this lower bound, all pure stationary schedulers $\sched$ with $\probmeasure{\mdpwithinitstate{\mdp}{\state}}{\sched}(\eventually \states \setminus \submdp{\states}{\stateactionpairset})$ have to be considered.
	To obtain such a lower bound, we can provide a lower bound on the probability of some finite path that leaves the EC.
	Since there has to be such a path that visits each $\state' \in \submdp{\states}{\stateactionpairset}$ at most once, we compute $p$ as follows:
	\[
	p = \prod_{\state' \in \submdp{\states}{\stateactionpairset}} \min_{\action \in \act{\state'}} \min_{\state'' \in \supp{\tuple{\state',\action}}} \transitions(\state',\action,\state'')
	\]
	\item Set all transition probabilities $\transitions(\state,\action,\state')$ with $\tuple{\state,\action} \in\stateactionpairset$ to 0.
	\item For each $\state' \in \submdp{\states}{\stateactionpairset}$ and $\action \in \act{\state'}$ such that $\tuple{\state',\action} \notin \stateactionpairset$, add a fresh action $\action'$ and set for each $\state'' \in \states $:
\[
	\transitions(\state,\action',\state'') = p \cdot \transitions(\state',\action,\state'') + (1-p) \cdot \iverson{\state = \state''}
	.
\]
\end{enumerate}
With these steps, we have eliminated all end components consisting of states in $\statesmaybe = \states \setminus \stateszero$. However, each state $s$ within an end component gets an additional self loop probability of $(1-p)$, which is an upper bound on the probability that we cycle through the end component and visit $s$ again.
Therefore, performing the approach of~\cite{BKLPW17} yields the desired values.

We remark that computing the values $\maxflow{\state}$ as above can lead to very large values which affect the numerical stability of the $\milp$ solving.
However, for the models in our experiments in \Cref{sec:evaluation}, this was not a concern.

\section{Infinite Rewards}
\label{app:B}

We now consider $\psma$ instances where infinite expected reward can be collected, \ie for state $\state$, objective $\objj$ and $\sched \in \psscheds{\mdp}$ we potentially have $\expval{\mdp_\state}{\sched}(\rewstructj\eventually\goalstatesj) = \infty$.
We treat such instances by a combination of preprocessing steps and (slight) modifications of the constraints in~\Cref{lp:unichainfinite,lp:endcomponents}.
See~\Cref{app:C} for the encoding of \Cref{sec:achievability:altencoding}.
As before, let
\[
\statesinf = \set{\state \in \states \mid  \fa{\sched \in \psscheds{\mdp}} \ex{\objindex} \objj \text{ is minimizing  and }\expval{\mdp_\state}{\sched}(\rewstructj\eventually\goalstatesj) = \infty }
\]
be the set of states for which all schedulers induce infinite reward with respect to at least one minimizing objective.
We can determine $\statesinf$ by first finding the end components of $\mdp$ in which no reward for any minimizing objective is collected. $\statesinf$ corresponds to the set of states that can not reach such an end component.

We can show for all $\sched \in \psscheds{\mdp}$ that $\probmeasure{\mdp}{\sched}(\eventually \statesinf) > 0$ implies $\expval{\mdp}{\sched}(\rewstructj\eventually\goalstatesj) = \infty$ for at least one minimizing objective $\objj$.
Hence, such a scheduler does not achieve the given point $\point$.
We assume $\sinit \notin\statesinf$ (otherwise $\psach{\mdp}{\multiobjquery} = \emptyset$).
To exclude schedulers $\sched \in \psscheds{\mdp}$  with $\probmeasure{\mdp}{\sched}(\eventually \statesinf) > 0$, we consider the sub-MDP $\submdp{\mdp}{\stateactionpairset_\mathrm{fin}, \sinit}$ instead of $\mdp$, where $\stateactionpairset_\mathrm{fin}$ is the largest subset of $(\states \setminus \statesinf) \times \act{}$ that is closed for $\mdp$. 
The achievable points for $\submdp{\mdp}{\stateactionpairset_\mathrm{fin}, \sinit}$ coincide with the achievable points for $\mdp$. 
Moreover, there is a scheduler for $\submdp{\mdp}{\stateactionpairset_\mathrm{fin}, \sinit}$ that  for each minimizing objective induces a finite expected reward at every state.
This is a requirement for the existence of a (real-valued) solution of the constraints in~\Cref{lp:unichainfinite}.

$\submdp{\mdp}{\stateactionpairset_\mathrm{fin}, \sinit}$ can still contain ECs in which infinite reward is collected for either a minimizing or a maximizing objective.
However, the constraints in \Cref{lp:endcomponents} should only apply to ECs without any rewards.
Hence, the constraints are considered for every maximal EC $\stateactionpairset \in \mecs{\submdp{\mdp}{\stateactionpairset_{\cap}}}$, where $\stateactionpairset_{\cap}$ is the largest subset of
$\stateactionpairset_\mathrm{fin} \cap \stateactionpairset_?^\objindex \cap \set{\tuple{\state,\action} \mid \fa{\state'} \rewstructj(\state,\action,\state')=0 }$
that is closed for $\mdp$.

\Cref{sec:achievability:unichain} considers upper bounds $\uppervalueboundsj \in \qq$ for the maximal expected rewards at state $\state$ with respect to $\objj$.
We consider the case where this value is infinite by computing sufficiently large bounds as follows:
Let $\mdp'$ be the MDP obtained by eliminating the end components in which a positive reward can be collected using a very similar construction as in \Cref{app:A}.
Note that we can not collect infinite reward in $\mdp'$.
We can show that the expected rewards for $\mdp'$ are an upper bound for the expected rewards for $\mdp$, assuming that  only strategies yielding finite rewards are considered.

For maximizing objectives $\objj$, we also have to allow strategies that do collect infinite reward.
We therefore add additional constraints that detect if infinite reward is collected.
The idea is to compute the probability that a state in $\stateszeroj \cup \{ \state \mid \lpsolof{\lpecvars} = 1\}$ is reached.
If this probability is below 1, infinite reward is collected.
An additional binary variable $b^j$ is added, to ensure that either infinite reward is collected or the threshold given by $\tupleaccess{\point}{\objindex}$ is satisfied.
The additional constraints are shown in \Cref{lp:infinite}.
Observe that strict inequalities as in Line~\ref{lp:inf:strict} are not allowed in \milp encodings.
 However, we can replace constraints of the form $a < b$ by $a + \epsilon \le b$ and ask for a solution that maximizes $\epsilon$.
\begin{figure}[t]
	\begin{lp}
	\lpline{}{\fa{\text{ maximizing } \objj \text{ with infinite reward possible}}}{}{}{}{}
	\lpline{"}{}{b^\objindex}{\in \set{0,1}}{}{}
	\lpline{"}{}{\lpvalvar{\sinit}{\objindex}}{\ge \tupleaccess{\point}{\objindex} \cdot b^j}{}{}
	\lpline{"}{}{w_{\sinit}^\objindex}{< 1 + b^j}{}{lp:inf:strict}
	\lpline{"}{\fa{\state \in \stateszeroj}}{w_s^\objindex}{= 1}{}{}
	\lpline{"}{\fa{\state \in \statesmaybej}}{}{}{}{}
	\lpline{""}{\fa{\action \in \act{\state}}}{w_{s,\action}^\objindex}{\ge 1 - \lpactionvarsa}{}{}
	\lpline{"""}{}{w_{s,\action}^\objindex}{\ge \lpecactionvarsaj}{}{}
	\lpline{"""}{}{w_{s,\action}^\objindex}{\ge \sum_{\state' \in \statesmaybej} \transitions(\state,\action,\state') \cdot w_{s'}^\objindex}{}{}
	\lpline{""}{}{w_{s}^\objindex}{= \sum_{\action \in \act{\state}} w_{s,\action}^\objindex - (|\act{\state}|-1)}{}{}
	\end{lp}
	\caption{\milp encoding for maximizing objectives with possibly infinite rewards.}
	\label{lp:infinite}
\end{figure}

\begin{theorem}
	With the modifications as above, the constraints in \Cref{lp:unichainfinite,lp:endcomponents} applied to $\submdp{\mdp}{\stateactionpairset_\mathrm{fin}, \sinit}$ are feasible iff	$\point \in \psach{\mdp}{\multiobjquery}$.
\end{theorem}
\section{Extensions for Alternative Encoding}
\label{app:C}
The alternative encoding can be lifted to multichain MDP as well. For this, we use the constraints from \Cref{lp:endcomponents} without Line~\ref{lp:ec:val} in conjunction with the constraints in \Cref{app:lp:alternative}.
The latter is a slight extension of the encoding from~\Cref{lp:alternative}.
It considers additional variables $\lpflowecvars$ which can only be non-zero, if $\state$ lies on an EC. This idea is similar to the variables $\lpecflowvarsaj$ in~\Cref{lp:endcomponents}.

After performing the preprocessing steps from \Cref{app:B}, the encoding also supports infinite rewards for minimizing objectives.
In incorporation of infinite rewards for maximizing objectives is left for future work.
\begin{figure}[t]
	\begin{lp}
	\lpline{}{\fa{\state \in \states}}{}{}{Select an action at each state}{}
	\lpline{"}{\fa{\action \in \act{\state}\enlargecolumn}\enlargecolumn}{\lpactionvarsa}{\in \set{0,1}}{}{}
	\lpline{"}{}{\sum_{\action \in \act{\state}}\lpactionvarsa}{=1}{}{}
	\lpline{}{\fa{\state \in \statesmaybe}\enlargecolumn}{\lpflowecvars}{\in [0, \maxflow{\state}]}{}{}
	\lpline{"}{}{\lpflowecvars}{\le \iverson{\ex{\stateactionpairset \in \mecs{\submdp{\mdp}{\stateactionpairset_?}}} \state \in \submdp{\states}{\stateactionpairset}} \cdot \lpecvars}{}{}
	\lpline{"}{\fa{\action \in \act{\state}}\enlargecolumn}{\lpflowvarsa}{\in [0, \maxflow{\state} \cdot \lpactionvarsa]}{}{}
	\lpline{"}{~~~~~~}{\lpflowecvars + \sum_{\action \in \act{\state}} \lpflowvarsa}{=\iverson{\state = \sinit} + \sum_{\tuple{\state',\action'} \in \predecessors{\state}} \transitions(\state',\action',\state) \cdot \lpflowvar{\state'}{\action'}}{}{}
	\lpline{}{}{1}{=\sum_{\state \in \statesmaybe} \Big( \lpflowecvars +  \sum_{\action \in \act{\state}} \lpflowvarsa\cdot  \sum_{\state' \in \stateszero} \transitions(\state,\action,\state') \Big)}{}{}
	\lpline{}{\fa{\objindex \in \set{1, \dots, \numobj}}\enlargecolumn}{}{}{}{}
	\lpline{"}{}{\lpvalvar{\sinit}{\objindex}}{= \sum_{\state \in \statesmaybe} \sum_{\action \in \act{\state}} \lpflowvarsa \cdot \sum_{\state' \in \states} \left(\transitions(\state,\action,\state') \cdot  \rewstructj(\state, \action, \state')\right) }{}{}
	\lpline{"}{}{\lpvalvar{\sinit}{\objindex}}{\objrelj \tupleaccess{\point}{\objindex}}{}{}
	\end{lp}
	\caption{\milp encoding for total reward objectives on multichain MDP.}
	\label{app:lp:alternative}
\end{figure}

\section{Details for Bounded Memory Achievability}\label{app:D}
\subsection{Pure strategies encoded by Mealy machines}
A pure strategy~$\sched$ for $\mdp = \mdptuple$ can be encoded by a \emph{Mealy machine} 
$\schedtuple$ where~$\memorystates$ is a finite set
of memory states,~$\minit \in \memorystates$ the \emph{initial memory state},
$\schednextaction$ the \emph{next action function}
$\schednextaction\colon \states \times \memorystates \rightarrow \actions$ where
$\schednextaction(\state,\memorystate) \in \act{\state}$ for any~$\state\in \states$ and~$\memorystate \in \memorystates$,
and $\schedmemoryupdate$ the \emph{memory update function}
$\schedmemoryupdate \colon \memorystates \times \states \times \actions \rightarrow \memorystates$.  A strategy is
\emph{$K$-memory} if~$|\memorystates| = K$.

The strategy $\sched$ induces an MC $\mc{\mdp}{\sched}$ defined on the state space $\states \times \memorystates$ with initial state $(\sinit, \minit)$ such that, for any pair of states $(\state,\memorystate)$ and~$(\state',\memorystate')$, the probability of transition $(\state,\memorystate)$ to $(\state',\memorystate')$ when choosing action $\action$ is 
equal to $\transitions(\state,\action,\state') \cdot \iverson{\action = \schednextaction(\state,\memorystate)} \cdot \iverson{\memorystate' = \schedmemoryupdate(\memorystate,\state, \action)}$.
\begin{remark}
	Let $K' \leq K$, from every $\sched' \in \pschedsK{\mdp}{K'}$, we can trivially construct a strategy $\sched \in \pschedsK{\mdp}{K}$ with $K {-} K'$ unused memory states.
	Thus, $\pKach{\mdp}{\multiobjquery}{K'} \subseteq \pKach{\mdp}{\multiobjquery}{K}$.
	Therefore, to reason about points achieved by pure strategies with memory of maximal size $K$, it is sufficient to only consider pure $K$-memory strategies.
\end{remark}

\subsection{Proof of Lemma~\ref{lem:bounded-mem-equiv}}\label{app:proof-bounded-mem-equiv}
We left out the proof of Lemma~\ref{lem:bounded-mem-equiv} from the main text: we present it here. Recall the statement:

\schedEquivLemma*

\begin{proof}
	Recall that $\states' = \states \times \memorystates$, $\actions' = \actions \times \memorystates$, and $\sinit' = \left ( \sinit, \minit \right )$.
	First, from a strategy $\sched \in \pschedsK{\mdp}{K}$ encoded as $\tuple{\memorystates', \schednextaction, \schedmemoryupdate, \minit'}$, we can construct a strategy $\sched' \in \psscheds{\mdpndmemory{K}}$ as follows.
	Since $\sched$ is $K$-memory, we link each memory state of $\sched$ to a memory state of $\ndmemory{K}$. So, consider w.l.o.g. that $\memorystates' = \memorystates$ and $\minit' = \minit$.
	For all $\left (\state, \memorystate\right ) \in \states \times \memorystates$, let $\action = \schednextaction(\state, \memorystate)$ and $m' = \schedmemoryupdate\left ( \memorystate, \state, \action \right )$, we define
	$\sched'\left (\state, \memorystate \right )= \left (\action,\memorystate' \right )$.
	Note that $\sched'$ is well defined using completeness of $\ndmemory{K}$: if $\action \in \actions(s)$, then $(\action, \memorystate') \in \actions'(\state, \memorystate)$ since $\ndmemoryupdate(\memorystate) = M$.
	
	Second, we can also construct a strategy $\sched = \tuple{\memorystates', \schednextaction, \schedmemoryupdate, \minit'}$ from any strategy $\sched' \in \psscheds{\mdpndmemory{K}}$ as follows. Let $\memorystates' = \memorystates$, and $\minit' = \minit$, for all $\state \in \states$ and $\memorystate \in \memorystates$ such that $\sched'(\state, \memorystate) = (\action, \memorystate')$, we define $\schednextaction(\state, \memorystate) = \action$, $\schedmemoryupdate(\memorystate, \state, \action) = \memorystate'$, and  $\schedmemoryupdate(\memorystate, \state, \action') = \memorystate$ for $\action' \neq \action$.
	
	These constructions form an equivalence between pure $K$-memory strategies for $\mdp$ and pure stationary strategies for $\mdpndmemory{K}$.
	Let $\sched \in \pschedsK{\mdp}{K}$ and $\sched' \in \psscheds{\mdpndmemory{K}}$ be two such equivalent strategies, and let $\mc{\mdp}{\sched}$, $\mc{\left(\mdpndmemory{K}\right)}{\sched'}$ be the MCs induced respectively by $\sched$ and $\sched'$.
	Both MCs have the same state space defined on $\states \times \memorystates$ and initial state $(\sinit, \minit)$. Moreover, let $\mctransitions{\sched}$ be the transition function of $\mc{\mdp}{\sched}$ and $\mctransitions{\sched'}$ that of $\mc{\left(\mdpndmemory{K}\right)}{\sched'}$. We have 
	$\mctransitions{\sched}\left(\left(\state, \memorystate\right), \alpha, \left(\state', \memorystate'\right)\right) = \mctransitions{\sched'}\left(\left(\state, \memorystate\right), \left(\alpha, \memorystate' \right), \left(\state', \memorystate'\right)\right)$ for $\state, \state' \in \states$, $\action \in \actions$, and $\memorystate, \memorystate' \in \memorystates$.
	Modulo a bijection consisting in a renaming of actions, both MCs are thus identical.
	Consequently, the probability of paths and events in $\mc{\mdp}{\sched}$ and $\mc{\left(\mdpndmemory{K}\right)}{\sched'}$ are equal.
\end{proof}

\subsection{Pure Bounded-memory Pareto Approximation Problem}\label{app:pbp-problem}
\begin{nproblem}[framed]{Pure Bounded-memory Pareto Approximation Problem (\pbpapprox)}
	Input: & MDP $\mdp$, $\numobj$-dimensional multi-objective query $\multiobjquery$, memory bound $K \in \nn$, precision $\precision > \rr_{>0}^\numobj$.\\
	Output: & An $\precision$-approximation of $\closure{\multiobjquery}{\pKpareto{\mdp}{\multiobjquery}{K}}$.
\end{nproblem}
We already have all the ingredients to solve this problem: let $\mdp$ be an MDP and $\ndmemory{K}$ be a complete nondeterministic memory structure, recall Corollary~\ref{cor:bounded-to-stationnary}:
\corBoundedMultiQuery*
We can establish a reduction from \pbpapprox to \pspapprox: for inputs $\mdp$, $\multiobjquery$, and $K \in \nn$, we solve \pbpapprox by computing the solution of the \pspapprox problem for $\mdpndmemory{K}$ and $\multiobjquery'$ with the approach of Section~\ref{sec:pareto}.

\section{Evaluation with GLPK}
We repeated our experiments from \Cref{sec:evaluation} using the free \milp solver \tool{GLPK} instead of \tool{Gurobi}.
Since GLPK does not have native support for multi-threaded \milp solving, the benchmarks were run on a single core. Apart from that, we have the same setup as described in \Cref{sec:evaluation}.

\Cref{tab:resglpk} shows the results for all benchmarks under positional stationary strategies using the alternative encoding whenever possible (cf. \Cref{tab:res}).
For the second instance of \benchmark{jobs} we observed an internal \tool{GLPK} error.
We conjecture that this is a consequence of numerical inaccuracies.

\Cref{fig:results:scatterglpk} compares the runtimes of both encodings (cf. \Cref{fig:results:scatter}).
Finally, \Cref{fig:results:plotglpk} shows the resulting Pareto fronts for different kinds strategy classes (cf. \Cref{fig:results:plot}).
Note that we have not been able to obtain values for PM$_2$ and PS within reasonable time.
\begin{remark}
	When comparing the number of found Pareto optimal points ($|P|$ in \Cref{tab:res,tab:resglpk}) we can observe slight differences between \tool{Gurobi} and \tool{GLPK}.
	This is because, in general, the $\precision$-approximation that solves the pure stationary Pareto approximation problem is not unique and different MILP solvers might find different solutions.
\end{remark}
\begin{table}[h]
	\caption{Results for stationary strategies.}
	\label{tab:resglpk}
	\adjustbox{max width=\textwidth}{%
		\begin{tabular}{ll?crrr|rr|rr?crrr|rr|rr}
			\multicolumn{2}{l?}{Bench-} & \multicolumn{4}{c|}{Instance 1} & \multicolumn{2}{c|}{$\varepsilon {=} 0.01$} & \multicolumn{2}{c?}{$\varepsilon = 0.001$} & \multicolumn{4}{c|}{Instance 2} & \multicolumn{2}{c|}{$\varepsilon = 0.01$} & \multicolumn{2}{c}{$\varepsilon {=} 0.001$}\\
			\multicolumn{1}{c}{mark} & \multicolumn{1}{c?}{$\numobj$} 
			&Par. & \multicolumn{1}{c}{$|\states|$} & \multicolumn{1}{c}{$\%\stateactionpairset$} & \multicolumn{1}{c|}{$\overline{\act{}}$} & \multicolumn{1}{c}{Time} & \multicolumn{1}{c|}{$|\points|$} &   \multicolumn{1}{c}{Time} & \multicolumn{1}{c?}{$|\points|$}  
			&Par. & \multicolumn{1}{c}{$|\states|$} & \multicolumn{1}{c}{$\%\stateactionpairset$} & \multicolumn{1}{c|}{$|\act{}|$} & \multicolumn{1}{c}{Time} & \multicolumn{1}{c|}{$|\points|$} &   \multicolumn{1}{c}{Time} & \multicolumn{1}{c}{$|\points|$}  
			\\\thickhline
			\benchmark{dpm} & 2$^*$ & 2 & 1272 & 32 & 3.2 & 1428 & 40 & \multicolumn{2}{c?}{\timeout} & 3 & 1696 & 30 & 3.2 & \multicolumn{2}{c|}{\timeout} & \multicolumn{2}{c}{\timeout}\\
\benchmark{eajs} & 2$^*$ & 2-3 & 689 & 0 & 1.2 & 7 & 23 & 1249 & 201 & 3-6 & $2{\cdot}10^4$ & 0 & 1.2 & \multicolumn{2}{c|}{\timeout} & \multicolumn{2}{c}{\timeout}\\
\benchmark{jobs} & 3$^*$ & 3-2 & 17 & 0 & 1.1 & 2 & 3 & 3 & 3 & 5-2 & 117 & 0 & 1.5 & \multicolumn{2}{c|}{\error} & \multicolumn{2}{c}{\error}\\
\benchmark{mutex} & 3$^*$ & 1 & 1795 & 36 & 2.2 & \multicolumn{2}{c|}{\timeout} & \multicolumn{2}{c?}{\timeout} & 2 & $1{\cdot}10^4$ & 33 & 2.3 & \multicolumn{2}{c|}{\timeout} & \multicolumn{2}{c}{\timeout}\\
\benchmark{polling} & 2 & 2-2 & 233 & 86 & 1.5 & \multicolumn{2}{c|}{\timeout} & \multicolumn{2}{c?}{\timeout} & 3-2 & 990 & 84 & 1.8 & \multicolumn{2}{c|}{\timeout} & \multicolumn{2}{c}{\timeout}\\
\benchmark{rg} & 2$^*$ & 2-1-20 & 2173 & 14 & 2.9 & \multicolumn{2}{c|}{\timeout} & \multicolumn{2}{c?}{\timeout} & 5-2-50 & $3{\cdot}10^4$ & 5 & 3.1 & \multicolumn{2}{c|}{\timeout} & \multicolumn{2}{c}{\timeout}\\
\benchmark{rover} & 2$^*$ & 2500 & $2{\cdot}10^4$ & 0 & 1.2 & \multicolumn{2}{c|}{\timeout} & \multicolumn{2}{c?}{\timeout} & 5000 & $4{\cdot}10^4$ & 0 & 1.2 & \multicolumn{2}{c|}{\timeout} & \multicolumn{2}{c}{\timeout}\\
\benchmark{serv} & 2$^*$ &  & $5{\cdot}10^4$ & 93 & 1.9 & \multicolumn{2}{c|}{\timeout} & \multicolumn{2}{c?}{\timeout} &   &   &   &   &   &   &   &  \\
\benchmark{str} & 2$^*$ & 30 & 1426 & 0 & 1.3 & 11 & 21 & \multicolumn{2}{c?}{\timeout} & 500 & $4{\cdot}10^5$ & 0 & 1.3 & 3006 & 17 & \multicolumn{2}{c}{\timeout}\\
\benchmark{team2} & 2$^*$ & 2 & 1847 & 24 & 1.2 & 2 & 5 & 2 & 5 & 3 & $1{\cdot}10^4$ & 21 & 1.2 & 36 & 38 & \multicolumn{2}{c}{\timeout}\\
\benchmark{team3} & 3$^*$ & 2 & 1847 & 24 & 1.2 & 49 & 15 & \multicolumn{2}{c?}{\timeout} & 3 & $1{\cdot}10^4$ & 21 & 1.2 & \multicolumn{2}{c|}{\timeout} & \multicolumn{2}{c}{\timeout}\\
\benchmark{uav} & 2$^*$ & 750 & $2{\cdot}10^5$ & 29 & 1.6 & \multicolumn{2}{c|}{\timeout} & \multicolumn{2}{c?}{\timeout} & 1000 & $4{\cdot}10^5$ & 31 & 1.8 & \multicolumn{2}{c|}{\timeout} & \multicolumn{2}{c}{\timeout}\\
\benchmark{wlan} & 2$^*$ & 0 & 2954 & 0 & 1.3 & \multicolumn{2}{c|}{\timeout} & \multicolumn{2}{c?}{\timeout} & 2 & $3{\cdot}10^4$ & 0 & 1.3 & \multicolumn{2}{c|}{\timeout} & \multicolumn{2}{c}{\timeout}
		\end{tabular}%
	}
\end{table}

\begin{figure}[h]
	\centering
	\setlength{\plotsize}{0.42\textwidth}
	\begin{subfigure}[b]{\plotsize}
		\centering
		\begin{tikzpicture}
		\path[use as bounding box] (-0.6,-1.1) rectangle (4.5,3.5);
		\begin{axis}[
		width=\plotsize,
		height=\plotsize,
		axis equal image,
		xmin=1,
		ymin=1,
		ymax=40000,
		xmax=40000,
		xmode=log,
		ymode=log,
		axis x line=bottom,
		axis y line=left,
		xtick={1,6,60,600,6000},
		xticklabels={1,6,60,600,6000},
		extra x ticks = {16000,32000},
		extra x tick labels = {{\timeout/\memout/\error},N/S},
		extra x tick style = {grid = major},
		ytick={1,6,60,600,6000},
		yticklabels={1,6,60,600,6000},
		extra y ticks = {16000},
		extra y tick labels = {\timeout/\memout/\error},
		extra y tick style = {grid = major},
		xlabel style={yshift=0.cm,xshift=-0.4cm},
		ylabel style={yshift=-0.43cm,xshift=-0.cm},
		yticklabel style={font=\scriptsize},
		xticklabel style={rotate=290,anchor=west,font=\scriptsize},
		legend pos=south east,
		legend columns=-1,
		legend style={nodes={scale=0.75, transform shape},inner sep=1.5pt,yshift=0.1cm,xshift=-0.3cm},
		]
		\addplot[
		scatter,
		only marks,
		scatter/classes={
			low={mark=square*,blue,mark size=1.25},
			high={mark=diamond*,orange,mark size=1.75}
		},
		scatter src=explicit symbolic
		]%
		table [col sep=semicolon,x=alt,y=classic,meta=type] {glpkscatterdata.csv};
		\legend{{$\varepsilon {=} 0.01$},{$\varepsilon {=} 0.001$}};
		\addplot[no marks] coordinates {(0.01,0.01) (16000,16000) };
		\addplot[no marks, densely dotted] coordinates {(0.01,0.1) (1600,16000)};
		\end{axis}
		\end{tikzpicture}
		\caption{Alt.\ ($x$) vs.\ original ($y$) encoding.}
		\label{fig:results:scatterglpk}
	\end{subfigure}
	\begin{subfigure}[b]{1.2\plotsize}
		\centering
		\begin{tikzpicture}
		\path[use as bounding box] (-0.6,-1.1) rectangle (5.6,3.5);
		\begin{axis}[
		width=1.4\plotsize,
		height=\plotsize,
		xmin=0.163,
		ymin=0.163,
		ymax=0.189,
		xmax=0.209,
		axis x line=bottom,
		axis y line=left,
		yticklabel style={font=\scriptsize},
		xticklabel style={font=\scriptsize},
		legend pos=north east,
		legend columns=1,
		legend style={nodes={scale=0.75, transform shape},inner sep=1.5pt,yshift=0cm,xshift=0cm},
		]

\addplot[blue,thick] table [col sep=semicolon,x=x,y=y] {polling/gen.csv};
\addplot[only marks, mark=o,red, thick] table [col sep=semicolon,x=x,y=y] {polling/glpkpm2.csv};
\addplot[only marks, mark=+,thick,orange] table [col sep=semicolon,x=x,y=y] {polling/glpkpmg.csv};
\addplot[only marks, mark=*,green!70!black, thick] table [col sep=semicolon,x=x,y=y] {polling/glpkps.csv};
		\legend{Gen,PM$_2$,PM$_\goalstates$,PS};
		\end{axis}
		\end{tikzpicture}
		\caption{Pareto fronts for restricted strategies.}
		\label{fig:results:plotglpk}
	\end{subfigure}
	\caption{Comparison of the two encodings (left) and impact of memory (right).}
\end{figure}
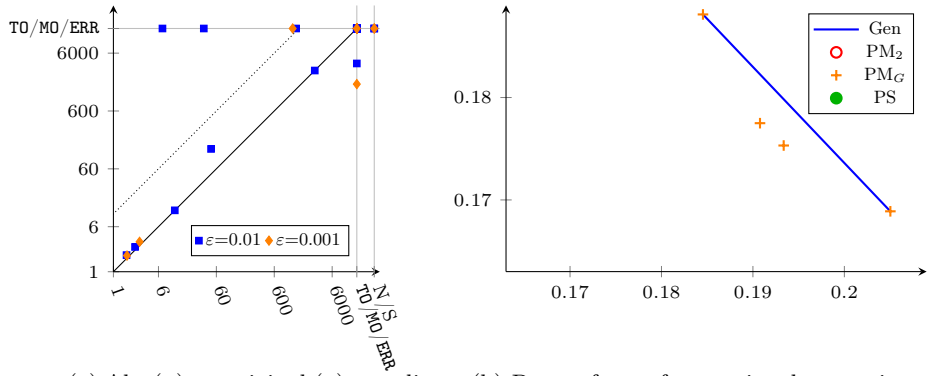

}{}


\vfill

{\small\medskip\noindent{\bf Open Access} This chapter is licensed under the terms of the Creative Commons\break Attribution 4.0 International License (\url{http://creativecommons.org/licenses/by/4.0/}), which permits use, sharing, adaptation, distribution and reproduction in any medium or format, as long as you give appropriate credit to the original author(s) and the source, provide a link to the Creative Commons license and indicate if changes were made.}

{\small \spaceskip .28em plus .1em minus .1em The images or other third party material in this chapter are included in the chapter's Creative Commons license, unless indicated otherwise in a credit line to the material.~If material is not included in the chapter's Creative Commons license and your intended\break use is not permitted by statutory regulation or exceeds the permitted use, you will need to obtain permission directly from the copyright holder.}

\medskip\noindent\includegraphics{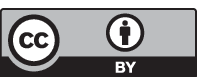}

\end{document}